\documentclass[a4paper,11pt]{article}
\usepackage[top=1 in,bottom=1 in,left=1.0 in,right=1.0 in]{geometry}
\usepackage{float}
\usepackage{cite}
\usepackage{graphicx}
\usepackage{subfigure}
\usepackage{amssymb,amsmath,amsthm,amscd}
\usepackage{stmaryrd}
\usepackage{mathrsfs}
\usepackage{tikz}
\usepackage[colorlinks,linkcolor=blue,citecolor=blue]{hyperref}

\usepackage{xcolor}

\def \h#1{\widehat{#1}}
\def \t#1{\widetilde{#1}}
\def \b#1{\overline{#1}}

\def \th#1{\widehat{\widetilde{#1}}}

\newtheorem{lemma}{Lemma}
\newtheorem{theorem}{Theorem}[section]

\newtheorem{prop}{Proposition}
\numberwithin{equation}{section}

\title{Squared eigenfunction symmetry of the D$\Delta$mKP hierarchy\\ and its constraint}
\author{Kui Chen,~Cheng Zhang,~Da-jun Zhang\footnote{Corresponding author. Email: djzhang@staff.shu.edu.cn}\\
{\small Department of Mathematics, Shanghai University, Shanghai 200444, P.R. China}}

\begin{document}

\maketitle

\begin{abstract}

In this paper squared eigenfunction symmetry of the differential-difference modified Kadomtsev-Petviashvili (D$\Delta$mKP) hierarchy
and its constraint are considered. Under the constraint, the Lax triplets of the D$\Delta$mKP hierarchy, together with their adjoint forms,
give rise to the positive relativistic Toda (R-Toda) hierarchy.
An invertible transformation is given to connect the positive and negative R-Toda hierarchies.
The positive R-Toda hierarchy is reduced to the differential-difference Burgers hierarchy.
We also consider another D$\Delta$mKP hierarchy and show that its squared eigenfunction symmetry constraint
gives rise to the Volterra hierarchy.
In addition, we revisit the Ragnisco-Tu hierarchy which is a squared eigenfunction symmetry constraint of the
differential-difference Kadomtsev-Petviashvili (D$\Delta$KP) system.
It was thought the Ragnisco-Tu hierarchy does not exist one-field reduction,
but here we find an one-field reduction to reduce the hierarchy to the Volterra hierarchy.
Besides, the differential-difference Burgers hierarchy are also investigated
in Appendix.  A multi-dimensionally consistent 3-point discrete Burgers equation
is given.

\begin{description}
\item[PACS numbers:] 02.30.Ik, 02.30.Ks, 05.45.Yv
\item[Keywords:] D$\Delta$mKP, squared eigenfunction symmetry constraint, relativistic Toda, one-field reduction,
Burgers
\end{description}
\end{abstract}

\section{Introduction}\label{sec-1}

It is common that an integrable system with a Lax pair usually has a squared eigenfunction symmetry
composed by the wave functions governed by the Lax pair and its adjoint form.
Such a symmetry is deeply related to $N$-soliton expression \cite{GGKM-CPAM-1974},
gradients of eigenvalues and nonlinearization of Lax pairs \cite{Cao-1990,Cao-NP-1990},
Mel'nikov-type integrable systems with self-consistent sources \cite{melnikov-1987} (also see \cite{Zhang-TMP-2010}), etc.
The  squared eigenfunction symmetries and their constraints have received intensively attention in the early 1990s
and many remarkable results were obtained, such as, by the symmetry constraint bridging a gap
between continuous (2+1)-dimensional and (1+1)-dimensional integrable systems \cite{cheng-pla-1991,Konopelchenko-pla-1991,Konopelchenko-Inverse-1991},
interpreting the squared eigenfunction symmetries as an assemble of isospectral flow symmetries \cite{satsuma-1990-kp},
understanding and solving the constrainted (2+1)-dimensional systems \cite{cheng-pla-1992,cheng-JPA-1992,cheng-jmp-1992,Cheng-CMP-1995,Dickey-CMP-1995,
Kundu-JMP-1995,LiuQP-PLA-1994,Loris-1997-jpa,
oevel-pa-1993,Oevel-Sch-RMP-1994,Oevel-Str-CMP-1993,XuBing-IP-1993,
XuBing-JPA-1992,Zhang-Ch-JMP-1994,Zhang-Ch-JPA-1994,ZhangYJ-JPA-1996},
and so on.

The research has also been extended in 1990s to differential-difference case with one independent discrete variable
\cite{Kajiwara-jmp-1991,Konopelchenko-jmp-1992,oevel-1996}
but the understanding did not go as far as the continuous case.
This is because at that time the discrete integrable systems were less understood than the continuous ones.
Recently, in  \cite{chen-2017-jnmp} it is shown that by the squared eigenfunction symmetry constraint
the differential-difference Kadomtsev-Petviashvili (D$\Delta$KP) system
is related to the (1+1)-dimensional differential-difference system, the Ragnisco-Tu hierarchy \cite{RT-1994,RT-1989},
which is a second discretization of the Ablowitz-Kaup-Newell-Segur(AKNS) system but different from
the Ablowitz-Ladik system \cite{al-1975}
by a different discretization of wave functions \cite{chen-2017-jnmp}.
The D$\Delta$KP hierarchy is related to the pseudo-difference operator $M$ \eqref{gt-kp-operator}
and the differential-difference modified Kadomtsev-Petviashvili (D$\Delta$mKP)
hierarchy is related to $L$ \eqref{mkp-hie-operator}, respectively.
In this paper we focus on the squared eigenfunction symmetry of the D$\Delta$mKP system,
and also revisit the D$\Delta$KP,
because the gauge connection between $M$ and $L$ will play a backstage role in the research.
As new results of the present paper, we obtain the following.
\begin{itemize}
\item{We reduce the Ragnisco-Tu hierarchy to the Volterra hierarchy,
which breaks the statement in \cite{RT-1994} that the hierarchy
``are essentially given by the non-existence of one-field reduction''. }
\item{We show that the squared eigenfunction symmetry constraint of the D$\Delta$mKP system
gives rise to the positive relativistic Toda (R-Toda) hierarchy in (1+1)-dimension.}
\item{The R-Toda($\pm$) hierarchies are unified by an invertible transformation.}
\item{The differential-difference Burgers hierarchy is obtained as reductions of the R-Toda($+$) hierarchy.}
\item{The D$\Delta$mKP(E) system that is related to $\bar L$ \eqref{mkp2hie-operator}
and its squared eigenfunction symmetry constraint not only gives rise to a decomposition of the Volterra hierarchy,
but also can be reduced to the later.}
\end{itemize}

This paper is organized as follows.
In Sec.\ref{sec-2} we introduce basic notations and derive the scalar D$\Delta$KP and D$\Delta$mKP hierarchies.
In Sec.\ref{sec-3} we revisit the squared eigenfunction symmetry constraint of the D$\Delta$KP
hierarchy and reduce the Ragnisco-Tu hierarchy to the Volterra hierarchy.
Then in Sec.\ref{sec-4} we deal with the D$\Delta$mKP system and show that its symmetry constraint
gives rise to the R-Toda hierarchy.
And in Sec.\ref{sec-5} we consider the D$\Delta$mKP(E) hierarchy and its  squared eigenfunction symmetry constraint.
Finally, conclusions are given in Sec.\ref{sec-6}.
There is an Appendix in which we will have a close look at the continuous and discrete Burgers equations.

\section{Preliminary}\label{sec-2}

\subsection{Notations}\label{sec-2-1}

Suppose that $u=u(n,x,t)$ and $v=v(n,x,t)$ are smooth functions of $(n,x,t)\in \mathbb{Z}\times \mathbb{R}^2$
and $C^{\infty}$ w.r.t. $(x,t)\in \mathbb{R}^2$.
Let $S[u]$ be a Schwartz space composed by all $f(u)$ that are $C^{\infty}$ differentiable w.r.t. $u$.
Here for two functions $f,g\in S$, the G$\hat{\mathrm{a}}$teaux derivative
of $f$ w.r.t. $u$ in direction $g$ is
\begin{align}\label{G-deriv}
f'[g]=\frac{d}{d \varepsilon}f (u+\varepsilon g)\Big|_{\varepsilon=0},
\end{align}
and a Lie product $\llbracket \cdot,\,\cdot \rrbracket$ is defined as
\begin{equation}
\llbracket f,g\rrbracket=f'[g]-g'[f].
\end{equation}
Without confusion usually we write $f(u)=f_n$.
We denote a shift by $E f_n=f_{n+1}$ and a difference by $\Delta f_n=(E-1)f_n=f_{n+1}-f_n$.
Note that there is an extended Leibniz rule for $\Delta$,
\begin{align}\label{dKP:Leib}
\Delta^{s}g_n=\sum^{\infty}_{i=0}\mathrm{C}_s^i\,(\Delta^ig_{n+s-i})\Delta^{s-i},\quad s\in \mathbb{Z},
\end{align}
where $\mathrm{C}_s^i$ is defined as
\begin{align}\label{def:Cij}
\mathrm{C}_s^i=\frac{s(s-1)(s-2)\cdots(s-i+1)}{i!}
\end{align}
 and $\mathrm{C}_0^0=1$.

In this paper we are interested in the following pseudo-difference operators
\begin{equation}\label{gt-kp-operator}
M=\Delta +u +u_1\Delta^{-1}+ u_2\Delta^{-2}+ \cdots,
\end{equation}
and
\begin{equation}\label{mkp-hie-operator}
  L=v\Delta+v_0+v_{1}\Delta^{-1}+\cdots,
\end{equation}
where $u, u_i\in S[u]$ and $v-1, v_i\in S[v]$.
For two functions $f_n,g_n\in S[u]\cup S[v]$, their inner product is defined as
\begin{equation}
<f_n,g_n> = \sum^{+\infty}_{n=-\infty}f_ng_n,
\end{equation}
by which we can define an adjoint operator of an operator $T$, denoted by $T^*$, via
\begin{equation}
<T f_n,g_n> = <f_n, T^* g_n>.
\end{equation}
$M^*$ and $L^*$ can be defined in this way.

\subsection{The scalar D$\Delta$KP hierarchy\cite{fu-2013-nonlinearity}}\label{sec-2-2}

Let us briefly recall the scalar D$\Delta$KP hierarchy derived from the Lax triplet composed by\footnote{Although in isospectral case
$x\equiv t_1$, when considering master symmetry which is in non-isospectral case, $x$ and $t_1$ must be considered different \cite{fu-2013-nonlinearity}.}
\begin{subequations}\label{gt-kp-spec}
  \begin{align}
  &  M \Psi =\lambda \Psi,\label{gt-kp-spec-t-x}\\
  &  \Psi_x = B_1 \Psi,  \label{gt-kp-spec-x}\\
  &  \Psi_{t_s} = B_s \Psi,~~s=1,2,\cdots, \label{gt-kp-spec-t}
  \end{align}
\end{subequations}
where $\lambda$ is a spectral parameter, $B_s=(M^s)_{\geq 0}$ contains
only non-negative $\Delta^j$ terms (i.e. $j\geq 0$) of $M^s$, e.g.,
\begin{subequations}\label{gt-kp-B1B2}
  \begin{align}
    B_1 & =  \Delta+ u, \\
    B_2 & =  \Delta^2 + (Eu+u)\Delta+ (\Delta u)+u^2+(Eu_1+u_1 ).
  \end{align}
\end{subequations}
The compatibility of \eqref{gt-kp-spec} leads to
\begin{subequations}
\begin{align}
& M_{t_s}=[B_s, M],\label{gt-kp-lax}\\
& M_{x}=[B_1, M],\label{gt-kp-lax-x}\\
& B_{1,t_s}-B_{s,x}+[B_1, B_s]=0,~~s=1,2,\cdots,\label{gt-kp-lax-t}
\end{align}
\end{subequations}
where $[A,B]=AB-BA$, and among which \eqref{gt-kp-lax-x} serves to express $u_i$ via $u$ by
\begin{subequations}\label{gt-kp-laxs1}
\begin{align}
   & \Delta u_1= u_{x}, \label{gt-kp-s1} \\
   & \Delta u_{s+1} = u_{s,x} - \Delta u_s - u u_s + \sum_{j=0}^{s-1} (-1)^j \mathrm{C}_{s-1}^{j}u_{s-j}\Delta^{j}E^{-s}u,~s=1,2,\cdots,
\end{align}
\end{subequations}
and \eqref{gt-kp-lax-t} provides a zero curvature representation for the scalar D$\Delta$KP hierarchy
\begin{equation}
u_{t_s}=B_{s,x}-[B_1, B_s],~~s=1,2,\cdots.
\label{kp-hie}
\end{equation}
This hierarchy can also be expressed as
\begin{equation}\label{gt-kp-hierar}
 u_{t_s} = \Delta \underset{\Delta}{\textrm{Res}} (M^s),~~s=1,2,\cdots.
\end{equation}
Denote the hierarchy by
\begin{equation}
u_{t_s}=G_s,
\end{equation}
in which the  D$\Delta$KP equation is  \cite{Date-1982,kanaga-1997}
\begin{equation}\label{gt-kp-kpeqn}
 u_{t_2}=G_2= (1+2\Delta^{-1}) u_{xx}+2uu_{x}- 2u_{x}.
\end{equation}
It can be proved that
\begin{equation*}
\llbracket G_j,G_s\rrbracket=0,
\end{equation*}
which means the flow $G_j$ is a symmetry of the whole  D$\Delta$KP hierarchy \eqref{kp-hie}.

Note that one may also alternatively rewrite $M$ in terms of $E$, denoted by
\begin{equation}\label{mkpequi-theo-equi-operatorkp}
     \bar{M}= E+ \bar{u} + \bar{u}_1E^{-1}+\cdots,
\end{equation}
where
\begin{equation}\label{u-u-bar}
     \bar{u}=u-1.
\end{equation}
In that case, $\bar{u}+1, \bar u_i \in S[u]$, and the Lax triplet
\begin{subequations}\label{mkpequi-theo-sp}
     \begin{align}
        & \bar{M} \Psi =\lambda \Psi,  \label{mkpequi-theo-sp-x} \\
        & \Psi_x =\bar{B}_1 \Psi,\\
        & \Psi_{t_s} = \bar{B}_s\Psi,~~s=1,2,\cdots, \label{mkpequi-theo-sp-t}
     \end{align}
\end{subequations}
leads to another scalar D$\Delta$KP hierarchy
\begin{equation}\label{mkpequi-theo-kphie}
     \bar{u}_{t_s} = \bar{G}_s,~~s=1,2,\cdots,
\end{equation}
which is the same as \eqref{kp-hie} under  \eqref{u-u-bar}.

\subsection{The scalar D$\Delta$mKP hierarchy}\label{sec-2-3}

Here we employ the Lax triplet approach to derive the scalar  D$\Delta$mKP hierarchy from the pseudo-difference operator \eqref{mkp-hie-operator}
that has been considered in \cite{Tamizhmani-2000}.
Let us start from the triplet
\begin{subequations}\label{mkp-hie-spec}
\begin{align}
& L \Phi=\lambda \Phi,  \label{mkp-hie-spec-1}\\
& \Phi_x= A_1\Phi, \label{mkp-hie-spec-x}\\
& \Phi_{t_s}=A_s\Phi,~~s=1,2,\cdots, \label{mkp-hie-spec-2}
\end{align}
\end{subequations}
where  $A_s=(L^s)_{\geq 1}$, e.g.,
\begin{subequations}\label{mkp-hie-A1A2}
  \begin{align}
     & A_1=  v\Delta, \label{mkp-hie-A1A2-A1} \\
     & A_2=  v(Ev)\Delta^2 + v( Ev_0+v_0+ \Delta v )\Delta,\label{mkp-hie-A1A2-A2}\\
     & \cdots\cdots .\nonumber
  \end{align}
\end{subequations}
The compatibility of \eqref{mkp-hie-spec} gives rise to
\begin{subequations}
\begin{align}
& L_{t_s}=[A_s, L],\label{mkp-hie-lax}\\
& L_{x}=[A_1, L],\label{mkp-hie-lax-x}\\
& A_{1,t_s}-A_{s,x}+[A_1, A_s]=0,~~s=1,2,\cdots. \label{mkp-hie-lax-t}
\end{align}
\end{subequations}
From \eqref{mkp-hie-lax-x} one can find
\begin{subequations}\label{lemma-rtoda-vab-v-proof}
   \begin{align}
      &  v \Delta v_0=v_x, \label{lemma-rtoda-vab-v-proof-1}\\
      &  vEv_1-(E^{-1}v)v_1=v_{0,x}-v\Delta v_0,  \label{lemma-rtoda-vab-v-proof-2}\\
      & vEv_{s+1}-(E^{-1-s}v)v_{s+1}=v_{s,x}-v\Delta v_s+\sum_{i=1}^{s}(-1)^{s+1-i}v_i E^{-s-1}\Delta^{s+1-i}v,  \label{lemma-rtoda-vab-v-proof-3}
   \end{align}
      \end{subequations}
which further yields expressions for  $v_s$ in terms of $v$, i.e.
\begin{subequations}\label{mkp-hie-v0v1}
  \begin{align}
     &  v_0= \Delta^{-1}(\ln v)_{x}, \label{mkp-hie-v0v1-v0} \\
     &  v_{1}=\frac{\Delta^{-2}(\ln v)_{xx}- \Delta^{-1} v_{x}}{(E^{-1}v)},\label{mkp-hie-v0v1-v1}\\
     & \cdots \cdots.\nonumber
  \end{align}
\end{subequations}
\eqref{mkp-hie-lax-t} serves as a zero curvature representation for the scalar D$\Delta$mKP hierarchy
\begin{equation}
v_{t_s} = K_s= (A_{s,x}-[A_1, A_s])\Delta^{-1},~~s=1,2,\cdots.
\label{mkp-hie}
\end{equation}
The first nonlinear equation of this hierarchy reads \cite{Tamizhmani-2000}
\begin{equation}\label{mkp-hie-mkp}
   (\ln v)_{t_2}= (1+2\Delta^{-1})(\ln v)_{xx} + (\ln v)_{x}(1+2\Delta^{-1})(\ln v)_{x}-2 v_{x},
\end{equation}
which is known as the D$\Delta$mKP equation and gives rise to the mKP equation in continuum limit.
Similar to \cite{fu-2013-nonlinearity},
one can prove that
\begin{equation}\label{mkp-hie-hierarchy}
    (\ln v)_{t_s} = \Delta \underset{\Delta}{\mathrm{Res}}(L^s \Delta^{-1}),~~s=1,2,\cdots,
 \end{equation}
and $K_j$ is a symmetry of the whole  D$\Delta$mKP hierarchy due to
\begin{equation*}
\llbracket K_j,K_s\rrbracket=0.
\end{equation*}

\subsection{Gauge equivalence}\label{sec-2-4}

Motivated by \cite{oevel-1996}, Ref.\cite{Tamizhmani-2000} considered gauge transformation of the pseudo-difference operators
\eqref{gt-kp-operator} and \eqref{mkp-hie-operator} by introducing an undetermined function $f_n$ such that
\begin{equation}\label{gt-kp-gaugeLM}
  f_nL=Mf_n.
\end{equation}
As a result one has
\begin{subequations}\label{gt-kp-vu}
\begin{align}
& v= \frac{f_{n+1}}{f_n},~~ v_0=\partial_x \ln f_n, \label{gt-kp-vu-a}\\
& u= \Delta^{-1}\partial_{x}\ln v +1-v,\label{MT}
\end{align}
\end{subequations}
where the later is considered as a Miura transformation to connect the D$\Delta$KP hierarchy $u_{t_s}=G_s$
and the D$\Delta$mKP hierarchy $v_{t_s}=K_s$, i.e. \eqref{kp-hie} and \eqref{mkp-hie}.
In addition, comparing \eqref{gt-kp-spec-t-x} and \eqref{mkp-hie-spec-1} we have
\begin{equation}\label{gt-kp-transsp}
  \Psi=f_n\Phi.
\end{equation}

Corresponding to the Lax triplets \eqref{gt-kp-spec} and \eqref{mkp-hie-spec}, their adjoint forms are
\begin{subequations}\label{gt-kp-adwave-kp}
  \begin{align}
     & M^* \Psi^*=\lambda \Psi^*,  \label{gt-kp-adwave-kp-1}\\
     & \Psi^*_{x}=-B_1^*\Psi^*, \label{gt-kp-adwave-kp-b}\\
     & \Psi^*_{t_s}=-B_s^*\Psi^*,~~s=1,2,\cdots, \label{gt-kp-adwave-kp-2}
  \end{align}
\end{subequations}
and
\begin{subequations}\label{gt-kp-adwave-mkp}
  \begin{align}
     & L^* \Phi^*=\lambda \Phi^*,  \label{gt-kp-adwave-mkp-1}\\
     & \Phi^*_{x}=-A_1^*\Phi^*,\label{gt-kp-adwave-mkp-x}\\
     & \Phi^*_{t_s}=-A_s^*\Phi^*,~~s=1,2,\cdots, \label{gt-kp-adwave-mkp-2}
  \end{align}
\end{subequations}
where $A^*$ and $B^*$ are adjoint operators of $A_s$ and $B_s$, and
$\Psi^*$ and $\Phi^*$ stands for solutions of \eqref{gt-kp-adwave-kp} and \eqref{gt-kp-adwave-mkp}.
Note that \eqref{gt-kp-adwave-kp} and \eqref{gt-kp-adwave-mkp} generate
the D$\Delta$KP hierarchy $u_{t_s}=G_s$ and the D$\Delta$mKP hierarchy $v_{t_s}=K_s$ as well.
Since $f_n$ is a scalar function, from the way to define adjoint operators, it is easy to get
\begin{equation}
L^* f_n =f_n M^*
\end{equation}
and consequently
\begin{equation}
\Phi^*= f_n \Psi^*.
\label{bt-kp-adj-trans}
\end{equation}

\section{Squared eigenfunction symmetry of the D$\Delta$KP}\label{sec-3}

In this section we first briefly revisit the squared eigenfunction symmetry of the D$\Delta$KP hierarchy and its constraint.
Then we present a new one-field reduction of the Ragnisco-Tu hierarchy.

\subsection{Squared eigenfunction symmetry}\label{sec-3-1}

It has been proved by use of additional symmetry in \cite{chen-2017-jnmp,he-2016} that
\begin{prop}
$(\Psi\Psi^*)_x$ is a symmetry of the whole D$\Delta$KP hierarchy \eqref{kp-hie},
provided $\Psi$ and $\Psi^*$ satisfy the Lax triplets \eqref{gt-kp-spec} and \eqref{gt-kp-adwave-kp}.
\end{prop}
In the following we revisit this result from the view point of $\tau$ function and present the following relation,
\begin{equation}\label{sigma-kp-flow}
\sigma=(\Psi\Psi^*)_x=\sum^{+\infty}_{s=1}G_s\lambda^{-s},
\end{equation}
where $G_s$ are the flows in the D$\Delta$KP hierarchy \eqref{kp-hie}.
To obtain the above relation, let us employ the results in \cite{adler-cmp-1999,Kajiwara-jmp-1991},
which includes some more explicit formulas of the 2DTL hierarchy
that was first systematically studied by the Sato's approach in the pioneer paper \cite{UT-2DTL-1984}.
In fact, the 2DTL hierarchy, which involves both pseudo-difference operators $\bar M$ and $\bar L^*$ (see \eqref{mkp2hie-operator}),
amounts to two  differential-difference KP hierarchies.
By the rational transformation (cf.\cite{Kajiwara-jmp-1991})
\begin{equation}\label{taukp-tau}
  u-1 = \partial_{x} \ln \frac{\tau_{n+1}}{\tau_{n}}
\end{equation}
where $\tau_n=\tau(n,x, t_1,t_2,\cdots)$,
the scalar D$\Delta$KP hierarchy \eqref{gt-kp-hierar} can be cast into bilinear forms \cite{UT-2DTL-1984,Kajiwara-jmp-1991,adler-cmp-1999}
\begin{equation}\label{taukp-bili}
  (D_{t_s} - p_s(\widetilde{\mathbf{D}} )) \tau_{n+1}\cdot\tau_n=0,~~t_1=x,~~s=2,3,\cdots
\end{equation}
where the $p_s(\mathbf{x})$ with $\mathbf{x}=(x_1,x_2,\cdots)$ are elementary Schur polynomials defined through
\begin{equation}\label{taukp-schur}
  \exp\,\biggl(\sum^{+\infty}_{j=1} x_j k^j\biggr) = \sum_{s=0}^{\infty} p_s(\mathbf{x})k^s,
\end{equation}
$\widetilde{\mathbf{D}}=(D_{t_1},\frac{1}{2}D_{t_2},\frac{1}{3}D_{t_3},\cdots)$
and $D_{x}$ is the Hirota bilinear operator defined as \cite{hirota-2004}
\begin{equation}\label{taukp-hirota}
  e^{h D_{x}} f(x)\cdot g(x) = f(x+h)g(x-h).
\end{equation}
The hierarchy \eqref{taukp-bili} can be alternatively written as
\begin{equation}\label{taukp-bili-z+-1}
  1+\sum_{s=1}^{\infty} \lambda^{-s} \partial_{t_s}\ln\frac{\tau_{n+1}}{\tau_{n}} = \frac{\tau_{n+1}(\mathbf{t}+[\lambda^{-1}])
  \tau_n(\mathbf{t}-[\lambda^{-1}])}{\tau_{n+1}\tau_{n}}
\end{equation}
with $\mathbf{t}=(t_1,t_2,\cdots)$ and  $[\lambda]=(\lambda,\lambda^2/2,\lambda^3/3,\cdots)$.
Then, taking derivative w.r.t. $x$ on  \eqref{taukp-bili-z+-1} and making use of  \eqref{taukp-tau},
it follows that
\begin{equation}\label{tau-Gstn}
\biggl( \frac{\tau_{n+1}(\mathbf{t}+[\lambda^{-1}])\tau_n(\mathbf{t}-[\lambda^{-1}])}
{\tau_{n+1}\tau_{n}}\biggr)_{x}=\sum_{s=1}^{\infty} \lambda^{-s} u_{t_s}.
\end{equation}

On the other hand, the eigenfunctions $\Psi$ and $\Psi^*$ can be expressed in terms of $\tau_n$ as the following \cite{Kajiwara-jmp-1991,adler-cmp-1999},
\begin{subequations}
\begin{align}
&   \Psi = \frac{\tau_n(\mathbf{t}-[\lambda^{-1}])}{\tau_n} e^{\sum_{s=1}^{\infty} t_s \lambda^s } \lambda^n, \label{taukp-theo-solu-psi}\\
&   \Psi^* = \frac{\tau_{n+1}(\mathbf{t}+[\lambda^{-1}])}{\tau_{n+1}} e^{-\sum_{s=1}^{\infty} t_s \lambda^s } \lambda^{-n}. \label{taukp-theo-solu-psicon}
\end{align}
\end{subequations}
Combining  \eqref{tau-Gstn}, \eqref{taukp-theo-solu-psi} and \eqref{taukp-theo-solu-psicon} together, we immediately reach
\begin{equation}\label{tau-psipsiconGs}
  (\Psi\Psi^*)_{x} = \sum_{s=1}^{\infty} \lambda^{-s} u_{t_s},
\end{equation}
i.e. \eqref{sigma-kp-flow}.
This means $\sigma=(\Psi\Psi^*)_x$ is not only a symmetry of the whole D$\Delta$KP hierarchy \eqref{kp-hie},
but also assembles all flows (isospectral symmetries) of the hierarchy.

\subsection{Symmetry constraint and the Ragnisco-Tu hierarchy}\label{sec-3-2}

Noting that $u_x$ is also a symmetry of the D$\Delta$KP hierarchy \eqref{kp-hie},
we consider the symmetry constraint $\sigma=u_x+(\Psi\Psi^*)_x=0$,
which leads to\footnote{In principle there is an integration constant $c$. Here we take $c=1$,
which is slightly different from \cite{chen-2017-jnmp} where $c$ is taken to be zero. }
\begin{equation}\label{rt-sckp}
  u=-\Psi \Psi^* +1.
\end{equation}
Taking
\begin{equation}
\Psi=Q_n,~~\Psi^*=R_n,
\label{Q-R}
\end{equation}
the pseudo-difference operator \eqref{gt-kp-operator} is written as (cf. \cite{chen-2017-jnmp})
\begin{equation}\label{rt-Mcompact}
  M=E -Q_n R_n -Q_n\Delta^{-1} R_n,
\end{equation}
and the spectral problem \eqref{gt-kp-spec-t-x} is cast into a matrix form
\begin{equation}\label{rt-spec}
  \left(
    \begin{array}{c}
      \psi_{1,n+1} \\
      \psi_{2,n+1} \\
    \end{array}
  \right)= \left(
             \begin{array}{cc}
               \lambda+Q_nR_{n} & Q_n \\
               R_{n} & 1 \\
             \end{array}
           \right) \left(
    \begin{array}{c}
      \psi_{1,n} \\
      \psi_{2,n} \\
    \end{array}
  \right).
\end{equation}
Meanwhile, \eqref{gt-kp-spec-x} and \eqref{gt-kp-adwave-kp-b} turn out to be
\begin{equation}\label{rt-H1}
  \left(
    \begin{array}{c}
      Q_{n} \\
      R_n \\
    \end{array}
  \right)_x= H_1=\left(
    \begin{array}{l}
      Q_{n+1}-Q_n^2R_n \\
      -R_{n-1}+Q_nR_n^2
    \end{array}
  \right).
\end{equation}
Note that the Ragnisco-Tu spectral problem \eqref{rt-spec} together with \eqref{rt-H1},
after replacing $R_n$ with $R_{n+1}$ and $\lambda$ with $2(\lambda-\beta)$,
is actually the Darboux transformation of the AKNS spectral problem (cf. \cite{AY-JPA-1994,Mik-TMP-2013}),
and \eqref{rt-spec} is gauge equivalent to the form \cite{chen-2017-jnmp}
\begin{equation}
  \left(
    \begin{array}{c}
      \psi_{1,n+1} \\
      -\psi_{2,n-1} \\
    \end{array}
  \right)= \left(
             \begin{array}{cc}
               \eta & Q_n \\
               R_{n} & -\eta \\
             \end{array}
           \right) \left(
    \begin{array}{c}
      \psi_{1,n} \\
      \psi_{2,n} \\
    \end{array}
  \right),
\end{equation}
which is a discretisation of the AKNS spectral problem in a way different from the Ablowitz-Ladik spectral problem \cite{al-1975}
\begin{equation}\label{AL}
  \left(
    \begin{array}{c}
      \phi_{1,n+1} \\
      \phi_{2,n+1} \\
    \end{array}
  \right)= \left(
             \begin{array}{cc}
               z & U_n \\
               V_{n} & 1/z \\
             \end{array}
           \right) \left(
    \begin{array}{c}
      \phi_{1,n} \\
      \phi_{2,n} \\
    \end{array}
  \right).
\end{equation}

The Ragnisco-Tu hierarchy is composed by \eqref{gt-kp-spec-t} and \eqref{gt-kp-adwave-kp-2}.
\begin{prop}\label{rt-theo-rthie}
Under the symmetry constraint \eqref{rt-sckp} together with \eqref{Q-R},  the system \eqref{gt-kp-spec-t} and \eqref{gt-kp-adwave-kp-2}
give rise to the recursive structure of the  Ragnisco-Tu hierarchy
  \begin{equation}\label{rt-theo-rthie-hie}
    \left(
      \begin{array}{c}
        Q_n\\
        R_n \\
      \end{array}
    \right)_{t_{s+1}} =H_{s+1}
    =M_R \left(
      \begin{array}{c}
        Q_n\\
        R_n \\
      \end{array}
    \right)_{t_s},~~
    H_1=\left(
    \begin{array}{l}
      Q_{n+1}-Q_n^2R_n \\
      -R_{n-1}+Q_nR_n^2
    \end{array}
  \right),~~s=1,2,\cdots,
  \end{equation}
  where the recursion operator $M_{R}$ is
  \begin{equation}\label{M_R}
    M_{R}=\left(
            \begin{array}{cc}
              1 & 0 \\
              0 & E^{-1} \\
            \end{array}
          \right) \bigg(\mu_n I - \left(
                                  \begin{array}{c}
                                    Q_n \\
                                    -R_{n+1} \\
                                  \end{array}
                                \right) (E+1)\Delta^{-1} ( R_{n+1} , Q_n ) \bigg)\left(
                                                                                     \begin{array}{cc}
                                                                                       E & 0 \\
                                                                                       0 & 1 \\
                                                                                     \end{array}
                                                                                   \right)
  \end{equation}
with $\mu_n=1+Q_nR_{n+1}$ and $I$ being the $2\times 2$ identity matrix.
\end{prop}
\begin{proof}
Introduce $ \mathcal{L} = M-1$.
Based on the results in \cite{chen-2017-jnmp}, one has the relation
    \begin{equation}
      \left(
        \begin{array}{l}
          (\mathcal{L}^s)_{\geq 0} \Psi \\
          ((\mathcal{L}^s)_{\geq 0})^*(-\Psi^*) \\
        \end{array}
      \right) = (M_R-1)^s \left(
                            \begin{array}{c}
                              \Psi \\
                              -\Psi^* \\
                            \end{array}
                          \right),~~s=1,2,\cdots
    \end{equation}
which leads to
    \begin{subequations}
      \begin{align*}
        \left(
          \begin{array}{l}
           Q_n \\
           R_n \\
          \end{array}
        \right)_{t_s}
         &=  \left(
               \begin{array}{l}
                 (M^s)_{\geq 0}Q_n \\
                 - [(M^s)_{\geq 0}]^*R_n \\
               \end{array}
             \right)= \left(
                        \begin{array}{l}
                           [(\mathcal{L}+1)^s]_{\geq 0}Q_n \\
                          -[[(\mathcal{L}+1)^s]_{\geq 0}]^*R_n
                        \end{array}
                      \right)
           \\
         &= \sum_{i=0}^{s} \mathrm{C}_{s}^{i} \left(
                                          \begin{array}{l}
                                            (\mathcal{L}^i)_{\geq 0} Q_n \\
                                            -[(\mathcal{L}^i)_{\geq 0}]^*R_n \\
                                          \end{array}
                                        \right)
           = \sum_{i=0}^{s} \mathrm{C}_{s}^{i}   (M_R-1)^i \left(
                            \begin{array}{c}
                              Q_n \\
                              -R_n \\
                            \end{array}
                          \right)\\
         &= (M_R)^s   \left(
                            \begin{array}{c}
                              Q_n \\
                              -R_n \\
                            \end{array}
                          \right).
      \end{align*}
    \end{subequations}
\end{proof}

\subsection{The Ragnisco-Tu hierarchy: Refined with asymptotic condition}\label{sec-3-3}

In the previous subsection, to derive the recursive structure  \eqref{rt-theo-rthie-hie} of the Ragnisco-Tu hierarchy,
we employed the results in \cite{chen-2017-jnmp}.
Note that the proof given in \cite{chen-2017-jnmp} has nothing with the asymptotic condition of $(Q_n,R_n)$ (i.e. $(\Psi,\Psi^*)$).

In the following, in order to meet the relation \eqref{rt-sckp},
we refine the Ragnisco-Tu hierarchy  under the asymptotic condition
\begin{equation}\label{asy-cond}
Q_{n+j}R_{n+i} \to -1,~~ (|n|\to +\infty),
\end{equation}
where $i,j$ are finite integers.
This will be helpful in discussing one-field reduction later.
Let us also introduce
\begin{equation}\label{H0'}
H'_0=H_{0}=(Q_n,-R_n)^T
\end{equation}
and define
\begin{equation}\label{Hs'}
H'_{s}=M_R^s H'_0,
\end{equation}
where $M_R$ is given as \eqref{M_R}.
With the asymptotic condition \eqref{asy-cond} it is easy to obtain
\begin{equation}\label{H1'}
H'_{1}=M_R H'_0=H_1+2H'_0
=\left(
    \begin{array}{l}
      Q_{n+1}-2Q_n-Q_n^2R_n \\
      -R_{n-1}+2R_n +Q_nR_n^2
    \end{array}
  \right).
\end{equation}
Write $H'_{s}=(H'_{s,1}, H'_{s,2})^T$, and by $[f_n]$ we denote
the residue terms of $f_n$ after modular $Q_{n+j}R_{n+i}$.
For example, $[Q_{n+1}R_{n-1}]=0$, $[Q_n^2R_n]=Q_n$ and $[Q_{n+j}]=Q_{n+j}$.
It is then easy to get
\[[H'_{0,1}]=Q_n,~~[H'_{1,1}]=Q_{n+1}-2Q_n-Q_n.\]
By observation of $[H'_{0,i}],~[H'_{1,i}]$ and the structure of the ``integration'' part of $M_R$,
we can find that under the asymptotic condition \eqref{asy-cond}
$[H'_{s,1}]$ is always a linear combination of $\{Q_{n+j}\}$ and $[H'_{s,2}]$ is always a linear combination of $\{R_{n+j}\}$.
This means, when we take into account of the asymptotic condition \eqref{asy-cond}
in deriving the refined Ragnisco-Tu hierarchy
\begin{equation}\label{rt-hie-refined}
(Q_n,R_n)^T_{t_s}=H'_s,~~s=0,1,\cdots,
\end{equation}
the flow $H'_s$ is actually certain linear combination of the flows $\{H_j\}$ with a form
\begin{equation}\label{H'-H}
H'_s=\sum^{s}_{j=0} c_j H_j,~~ c_s\equiv 1,~c_j\in \mathbb{Z}.
\end{equation}

\subsection{Reduced to the Volterra hierarchy}\label{sec-3-4}

In \cite{RT-1994} it is addressed that the Ragnisco-Tu hierarchy
``are essentially given by the non-existence of one-field reduction''.
In the following let us show how the Volterra hierarchy arises as a reduction from
the refined Ragnisco-Tu hierarchy \eqref{rt-hie-refined}.

Consider reduction
\begin{equation}\label{vrt-QR}
  R_{n+1}=-\frac{1}{Q_n}
\end{equation}
and introduce
\begin{equation}\label{vrt-rtVol}
  q_n =  \ln \frac{Q_n}{Q_{n-1}},
\end{equation}
which indicates that
\begin{equation}\label{vrt-conn}
 e^{q_n}= -Q_n R_n=\frac{Q_n}{Q_{n-1}}.
\end{equation}
By \eqref{vrt-QR} the two equations in  $(Q_n,R_n)^T_{t_1}=H'_1$ with \eqref{H1'}  are reduced to the same one,
\begin{equation}\label{vrt-Q-eq}
  (\ln Q_n)_{t_1}= \frac{Q_{n+1}}{Q_n}  + \frac{Q_{n}}{Q_{n-1}} -2,
\end{equation}
which gives rise to the well known Volterra equation
\begin{equation}\label{vrt-vol-eq}
  q_{n,t_1}= V_1=e^{q_{n+1}}-e^{q_{n-1}}.
\end{equation}

Next, for the recursive Ragnisco-Tu hierarchy \eqref{rt-theo-rthie-hie}, under the reduction \eqref{vrt-QR},
it is reduced to a scaler relation
\begin{equation}
    (\ln Q_{n})_{t_{s+1}} = ( 1+2\Delta^{-1} ) \Bigl(\frac{Q_{n+1}}{Q_{n}}E- E^{-1}\frac{Q_{n+1}}{Q_{n}} )
    (\ln Q_{n})_{t_{s}},~~s=1,2,\cdots.
  \end{equation}
Applying $E^{-1}\Delta$ on both sides, we immediately arrive at
\begin{equation}
    q_{n,t_{s+1}} = V_{s+1}= L^{}_V q_{n,t_{s}},~~ L^{}_V=( E+1 ) (e^{q_n}E- E^{-1}e^{q_n})\Delta^{-1},
\label{Volterra-hie}
\end{equation}
which is the recursive structure of the Volterra hierarchy (cf.\cite{suris-2003}).

Let us sum up the reduction results as follows.

\begin{theorem}\label{vrt-rthie-red}
Consider the relation \eqref{H'-H}, the refined  Ragnisco-Tu hierarchy \eqref{rt-hie-refined}
is reduced under the reduction \eqref{vrt-QR} to the Volterra hierarchy
\begin{equation}
q_{n,t_s}=V'_{s},~~ s=1,2,\cdots,
\label{Vol'-hie}
\end{equation}
where $q_n$ is defined as \eqref{vrt-rtVol}, $V'_1=V_1$ as defined in \eqref{vrt-vol-eq},
\begin{equation}\label{V'-V}
V'_s=\sum^{s}_{j=1} c_j V_j,~~ c_s\equiv 1,~c_j\in \mathbb{Z},
\end{equation}
and $V_j=L^{j-1}_V V_1$ with $L_V$ defined in \eqref{Volterra-hie}.
\end{theorem}


\section{Squared eigenfunction symmetry of the D$\Delta$mKP}\label{sec-4}

In this section we consider squared eigenfunction symmetry constraint of the D$\Delta$mKP system.
As a result, we will obtain the R-Toda hierarchy,
which further is reduced to two differential-difference Burgers hierarchies.

\subsection{Squared eigenfunction symmetry}\label{sec-4-1}

As for the squared eigenfunction symmetry of the D$\Delta$mKP hierarchy  \eqref{mkp-hie},
we have the following.

\begin{theorem}\label{tau-theo-squa}
The following relation holds
\begin{equation}\label{tau-psipsiconj-exp}
  (\Phi E\Delta^{-1} \Phi^*)_{x}= - \sum_{s=1}^{\infty} \lambda^{-s}   v_{t_s},
\end{equation}
where $\Phi$ and $\Phi^*$ satisfy the Lax triplets \eqref{mkp-hie-spec} and \eqref{gt-kp-adwave-mkp}.
\end{theorem}

\begin{proof}
We prove this Theorem by using the results of the  D$\Delta$KP hierarchy and the gauge equivalence relations.
Making use of relations \eqref{MT}, \eqref{gt-kp-transsp} and \eqref{bt-kp-adj-trans},
from \eqref{tau-psipsiconGs} we can find
\begin{align}
        (\Phi\Phi^*)_{x}= &(\Psi\Psi^*)_{x} = \sum_{s=1}^{\infty} \lambda^{-s} \partial_{t_s} u \nonumber\\
        = & \sum_{s=1}^{\infty} \lambda^{-s} (\Delta^{-1}\partial_{x}\ln v -v)_{t_s} \nonumber \\
        = & \Bigl( \Delta^{-1}\partial_{x}\frac{1}{v} -1\Bigr ) \sum_{s=1}^{\infty} \lambda^{-s} v_{t_s}.\label{eq0}
\end{align}
Note that from \eqref{mkp-hie-spec-x} and \eqref{gt-kp-adwave-mkp-x} we can derive a relation
\begin{equation}\label{tau-theo-squa-iden1}
    ( \Phi E \Delta^{-1}\Phi^* )_{x}= v \Delta (\Phi E \Delta^{-1}\Phi^* - \Phi\Phi^* ),
\end{equation}
from which we can find
\begin{equation}\label{eq1}
(\Delta^{-1}\partial_{x}\frac{1}{v} -1 )(\Phi E \Delta^{-1}\Phi^*)_{x} =-(\Phi\Phi^*)_{x}.
\end{equation}
Thus, combining   \eqref{eq0} and \eqref{eq1} we immediately reach \eqref{tau-psipsiconj-exp}
and complete the proof.
\end{proof}

Note that \eqref{tau-psipsiconj-exp} indicates $\sigma=(\Phi E\Delta^{-1} \Phi^*)_{x}$
provides a symmetry for the whole D$\Delta$mKP hierarchy \eqref{mkp-hie}.
Such a symmetry can also be constructed using the additional symmetry approach \cite{oevel-pa-1993} (also see \cite{cheng-2018}).

\subsection{Symmetry constraint and the R-Toda hierarchy}\label{sec-4-2}

\subsubsection{Spectral problem}\label{sec-4-2-1}

Since both $v_x$ and $(\Phi E\Delta^{-1} \Phi^*)_{x}$ are symmetries of the D$\Delta$mKP hierarchy, we consider the following
symmetry constraint
\begin{equation}\label{sc-mkp}
v=\Phi E\Delta^{-1} \Phi^*.
\end{equation}
For convenience, introduce
\begin{equation}\label{rtoda-spec-ab}
  a_n = \Phi,~~b_n = E \Delta^{-1}\Phi^*,
\end{equation}
under which the pseudo-difference operator $L$ is written as
\begin{equation}\label{rtoda-abL}
  L=a_nb_n\Delta + a_n\Delta^{-1}b_n\Delta.
\end{equation}
Note that \eqref{sc-mkp} and \eqref{rtoda-spec-ab} indicate
\begin{equation}
v=a_nb_n.
\label{v-ab}
\end{equation}

To prove the form \eqref{rtoda-abL}, we need to express all $v_s$ in terms of $a_n$ and $b_n$.

\begin{prop}\label{lemma-rtoda-vab}
 Under the symmetry constraint \eqref{sc-mkp}, all the $\{v_s\}_{s=0}^{\infty}$ defined by \eqref{lemma-rtoda-vab-v-proof}
 can be expressed in terms of $a_n$ and $b_n$ as the following,
 \begin{equation}\label{lemma-rtoda-vab-v}
   v_s=(-1)^s a_n E^{-1-s}\Delta^s b_n,~~s=0,1,2,\cdots,
 \end{equation}
where $a_n$ and $b_n$ are defined in \eqref{rtoda-spec-ab}.
\end{prop}

\begin{proof}
Under the constraint \eqref{sc-mkp}, the coupled system \eqref{mkp-hie-spec-x} and \eqref{gt-kp-adwave-mkp-x}
turns out to be
   \begin{equation}\label{lemma-rtoda-vab-v-proof-abx}
     a_{n,x}=a_nb_n(a_{n+1}-a_n),~~b_{n,x}=a_nb_n ( b_n-b_{n-1}).
   \end{equation}
This will be used to eliminate derivatives of $a_n$ and $b_n$ w.r.t. $x$ in \eqref{lemma-rtoda-vab-v-proof}.
Inserting \eqref{v-ab} and \eqref{lemma-rtoda-vab-v-proof-abx} into \eqref{lemma-rtoda-vab-v-proof-1} and \eqref{lemma-rtoda-vab-v-proof-2} respectively,
$v_0$ and $v_1$ are written as
     \begin{equation}\label{lemma-rtoda-vab-v-proof-v0v1ab}
       v_0=a_nb_{n-1},~~v_1= -a_n E^{-2}\Delta b_n.
     \end{equation}
Then we assume
\begin{equation}
     v_i=(-1)^i a_n E^{-1-i}\Delta^i b_n,~~i=0,1,\cdots,m,
   \end{equation}
by which the right hand side of \eqref{lemma-rtoda-vab-v-proof-3} with $s=m$ is expressed as
    \begin{equation}\label{lemma-rtoda-vab-v-proof-vs}
     a_nb_nE \big( (-1)^{m+1} a_n E^{-2-m}\Delta^m b_n \big)-(E^{-1-s}a_nb_n) (-1)^{m+1} a_n E^{-2-m}\Delta^m b_n.
    \end{equation}
By comparison we immediately from \eqref{lemma-rtoda-vab-v-proof-3} find
\begin{equation}
      v_{m+1}=(-1)^{m+1} a_n E^{-2-m}\Delta^m b_n.
\end{equation}
This means on basis of mathematical induction the expression \eqref{lemma-rtoda-vab-v}
is valid for all $s=0,1,\cdots$.

\end{proof}

Substituting  \eqref{lemma-rtoda-vab-v} into \eqref{mkp-hie-operator}
and make use of formula \eqref{dKP:Leib} with $s=-1$, one can arrive at \eqref{rtoda-abL}.
This indicates that the D$\Delta$mKP system is closed under the constraint \eqref{sc-mkp}.

Now, with \eqref{rtoda-abL} in hand, the spectral problem \eqref{mkp-hie-spec-1} turns out to be
\begin{equation}\label{rtoda-abspectral}
  \big(a_nb_n\Delta + a_n\Delta^{-1}b_n\Delta \big)\Phi = \lambda \Phi,
\end{equation}
which can be written as
\begin{equation}\label{rt-sp-1}
\left(
  \begin{array}{c}
    \phi_{1,n+1} \\
    \phi_{2,n+1} \\
  \end{array}
\right) = \left(
            \begin{array}{cc}
              {(\eta^2-z_n)}/{r_n} & -\eta/r_n  \\
              \eta & 0 \\
            \end{array}
          \right) \left(
                    \begin{array}{c}
                      \phi_{1,n} \\
                      \phi_{2,n} \\
                    \end{array}
                  \right),
\end{equation}
where $\phi_{1,n}=\Phi/a_n,~\lambda=\eta^2$, and
\begin{equation}\label{rtoda-abgt1abqr}
 z_n= -a_nb_n,~~r_n= a_{n+1}b_n,
\end{equation}
Here we specially note that we employ $z_n$ to avoid making confusion with $v$ used
in the pseudo-difference operator $L$, and also note that $z_n=-v$ in light of \eqref{v-ab}.

The spectral problem \eqref{rt-sp-1} can be further gauge transformed into
\begin{equation}\label{rtoda-absp2}
\left(
  \begin{array}{c}
    \phi'_{1,n+1} \\
    \phi'_{2,n+1} \\
  \end{array}
\right) = \left(
            \begin{array}{cc}
              \eta^2-z_n & -\eta  \\
              \eta \, r_n & 0 \\
            \end{array}
          \right) \left(
                    \begin{array}{c}
                      \phi'_{1,n} \\
                      \phi'_{2,n} \\
                    \end{array}
                  \right)
\end{equation}
by taking
\begin{equation}\label{rtoda-abgt2}
  (\phi_{1,n}, \phi_{2,n})=e^{-\Delta^{-1}\ln r_n} (\phi'_{1,n},  \phi'_{2,n}).
\end{equation}
After a further gauge transformation \cite{chen-2016}
 \begin{equation}
   \left(
     \begin{array}{c}
       \phi'_{1,n} \\
       \phi'_{2,n} \\
     \end{array}
   \right) = (-\eta)^n\left(
               \begin{array}{cc}
                 1 & 0 \\
                 0 & \frac{r_{n-1}}{\alpha} \\
               \end{array}
             \right)\Omega_n,~\eta=\frac{1}{\zeta},~z_n=1+\alpha R'_n,~r_n=\alpha^2Q'_n,
 \end{equation}
the spectral problem \eqref{rtoda-absp2} goes to the R-Toda spectral problem that reads \cite{suris-2003}
\begin{equation}\label{rtoda-sprtoda}
  \Omega_{n+1} = \left(
                   \begin{array}{cc}
                     \zeta(1+\alpha R'_n)-\zeta^{-1} & \zeta Q'_{n-1} \\
                     -\alpha & 0 \\
                   \end{array}
                 \right)\Omega_n,
\end{equation}
where $\alpha$ is an arbitrary constant.

Note that the one-component form of \eqref{rtoda-absp2}, which reads
\begin{equation}\label{rtoda-spuniversity}
  \phi'_{1,n+1} + z_n\phi'_{1,n}=\eta^2(\phi'_{1,n} -r_{n-1} \phi'_{1,n-1}),
\end{equation}
has been used as an unusual spectral problem (see Eq.(2.1) in \cite{Kharchev-1997-jmpa})
to study the R-Toda lattice;
\eqref{rtoda-absp2} was also restudied in \cite{Wen-RepMP-2011} without mentioning
the connection with the R-Toda lattice.

\subsubsection{The R-Toda hierarchy}\label{sec-4-2-2}

Under the constraint \eqref{sc-mkp}, the two equations \eqref{mkp-hie-spec-x} and \eqref{gt-kp-adwave-mkp-x}
in the Lax triplets \eqref{mkp-hie-spec} and \eqref{gt-kp-adwave-mkp}
can be cast into an evolution equation in terms of $(a_n,b_n)$ like \eqref{lemma-rtoda-vab-v-proof-abx}, i.e.
\begin{equation}\label{rtoda-ab-x}
    (\ln a_n)_{x} = b_n (a_{n+1} -a_n),~~(\ln b_n)_{x}=a_n(b_n-b_{n-1}),
\end{equation}
where $a_n$ and $b_n$ are given in \eqref{rtoda-spec-ab}.
To look at explicit forms of \eqref{mkp-hie-spec-2} and \eqref{gt-kp-adwave-mkp-2} in terms of $(a_n,b_n)$,
we need to use some recursive structures.

\begin{prop}\label{rtodap-prop-As}
 With the compact form \eqref{rtoda-abL} of $L$, $A_s$ allows the following two recursive relations
 \begin{subequations}\label{rtodap-prop-As-property}
   \begin{align}
      & A_{s+1} = L A_s +a_n b_n (E(L^s)_0)\Delta - a_n\Delta^{-1}(E\Delta^{-1}(A_s)^*E^{-1}\Delta b_n )\Delta,
      \label{rtodap-prop-As-property-1}\\
      & A_{s+1} = A_s L +(L^s)_0a_nb_n\Delta -(A_sa_n)\Delta^{-1}b_n\Delta,
      \label{rtodap-prop-As-property-2}
   \end{align}
 \end{subequations}
where  $(L^s)_0$ stands for the constant term of the operator $L^s$ with respect to $\Delta$.
 \end{prop}

\begin{proof}
By means of the identities
   \begin{equation}
       (A_s a_n \Delta^{-1} b_n\Delta )_{\leq 0} = (A_s a_n) \Delta^{-1}b_n\Delta,~~
      ~  \Delta^{-1} b_n\Delta A_s = \Delta^{-1} ( E\Delta^{-1} A_s^* E^{-1}\Delta b_n ) \Delta,
   \end{equation}
\eqref{rtodap-prop-As-property-1} and \eqref{rtodap-prop-As-property-2} can be derived from
assuming $A_{s+1}=(L^sL)_{\geq 1}$  and $A_{s+1}=(LL^s)_{\geq 1}$, respectively.
The detailed procedures are similar to the D$\Delta$KP case in \cite{chen-2017-jnmp}.
 \end{proof}

Now let us come to  \eqref{mkp-hie-spec-2} and \eqref{gt-kp-adwave-mkp-2}.

\begin{theorem}\label{rtodap-theo-hie}
\eqref{mkp-hie-spec-2} and \eqref{gt-kp-adwave-mkp-2} give rise to the recursive hierarchy
 \begin{equation}\label{rtodap-theo-hie-s}
    \left(
      \begin{array}{c}
       \ln a_{n} \\
       \ln b_{n} \\
      \end{array}
    \right)_{t_{s+1}} = L_{R}     \left(
      \begin{array}{c}
        \ln a_{n}\\
        \ln b_{n} \\
      \end{array}
    \right)_{t_{s}},~~s=1,2,\cdots,
  \end{equation}
where the initial member reads
  \begin{equation}\label{rtodap-theo-hie-1}
    (\ln a_n)_{t_1} = a_{n+1}b_n -a_nb_n,~~(\ln b_n)_{t_1}=a_nb_n-a_nb_{n-1},
  \end{equation}
and the recursion operator $L_{R}$ is
\begin{equation}
  L_{R} = \left(
            \begin{array}{cc}
              L_{R}^{11} & L_{R}^{12} \\
              L_{R}^{21} & L_{R}^{22} \\
            \end{array}
          \right)
\end{equation}
with elements
    \begin{align*}
       &L_{R}^{11} =b_n \Delta a_n +\Delta^{-1}b_n\Delta a_n + (\Delta a_n)b_nE\Delta^{-1},   \\
       &L_{R}^{12} =  \Delta^{-1}b_n (\Delta a_n) +  (\Delta a_n) b_nE\Delta^{-1},  \\
       &L_{R}^{21} =  E\Delta^{-1}a_n (E^{-1}\Delta b_n)+a_n (E^{-1}\Delta b_n )\Delta^{-1},     \\
       &L_{R}^{22} = -a_n  E^{-1}\Delta b_n + E\Delta^{-1}a_nE^{-1}\Delta b_n + a_n (E^{-1}\Delta b_n)\Delta^{-1}.
    \end{align*}
\end{theorem}

    \begin{proof}
      Based on \eqref{mkp-hie-hierarchy} and \eqref{v-ab}, one has
      \begin{equation}\label{L0}
      (L^s)_0 = \Delta^{-1}(\ln v)_{t_s} = \Delta^{-1}(\ln a_nb_n)_{t_s} =\Delta^{-1}\frac{1}{a_n} a_{n,t_s}+ \Delta^{-1}\frac{1}{b_n} b_{n,t_s}.
      \end{equation}
      Then, using  the recursive form  \eqref{rtodap-prop-As-property-1} and expression \eqref{rtoda-abL}, one has
\[ a_{n,t_{s+1}} = A_{s+1} a_n= (a_nb_n \Delta +a_n\Delta^{-1}b_n\Delta)a_{n,t_s} +  a_n b_n (E(L^s)_0)\Delta a_n+ a_n\Delta^{-1}b_{n,t_s}\Delta a_n,
\]
which, coupled with \eqref{L0}, gives rise to
\begin{equation*}
(\ln a_n)_{t_{s+1}} = L_{R}^{11} (\ln a_{n})_{t_s} + L_{R}^{12} (\ln b_{n})_{t_s}.
\end{equation*}
Similarly, using \eqref{rtodap-prop-As-property-2} and \eqref{L0} one can find
\begin{equation*}
(\ln b_n)_{t_{s+1}} = L_{R}^{21} (\ln a_{n})_{t_s} + L_{R}^{22} (\ln b_{n})_{t_s}.
\end{equation*}
\end{proof}

Further, after some calculations we can find that
\begin{theorem}\label{rtodap-coro-rtoda}
Under transformation \eqref{rtoda-abgt1abqr},  the hierarchy \eqref{rtodap-theo-hie-s} gives rise to the R-Toda(+) hierarchy
  \begin{equation}\label{rtodap-coro-rtoda-hie}
    \left(
      \begin{array}{c}
       \ln z_{n} \\
       \ln r_{n} \\
      \end{array}
    \right)_{t_{s+1}}
    = L_{R^+}  \left(
      \begin{array}{c}
       \ln z_{n} \\
       \ln r_{n} \\
      \end{array}
    \right)_{t_{s}}, ~~s=1,2,\cdots,
  \end{equation}
  where the first member is the R-Toda lattice (denoted by R-Toda$(+1)$)
  \begin{subequations}\label{rtodap-eq}
  \begin{align}
     & z_{n,t_1} = z_n (r_n- r_{n-1}), \\
     & r_{n,t_1} = r_n (r_{n+1}+z_{n+1}-r_{n-1}-z_{n}),
  \end{align}
\end{subequations}
  and the recursion operator $L_{R^+}$ reads
  \begin{equation}\label{rtodap-coro-rtoda-oper}
    L_{R^+} = \left(
               \begin{array}{cc}
                 z_n & (r_nE- E^{-1}r_n)\Delta^{-1} \\
                 (E+1)z_n & (Er_n E-E^{-1}r_n)\Delta^{-1}+r_n+ \Delta z_n \Delta^{-1} \\
               \end{array}
             \right).
  \end{equation}
\end{theorem}

\subsection{Further discussion on the R-Toda hierarchy}\label{sec-4-3}

\subsubsection{The R-Toda$(-)$ hierarchy}\label{sec-4-3-1}

The R-Toda(+) hierarchy \eqref{rtodap-coro-rtoda-hie} can be also derived
from the spectral problem \eqref{rtoda-absp2} by considering
the compatible condition with the time part
\begin{equation}\label{rtoda-absp2-time}
\left(
  \begin{array}{c}
    \phi'_{1,n} \\
    \phi'_{2,n} \\
  \end{array}
\right)_{t_s} = \left(
            \begin{array}{cc}
              A'_s & B_s'  \\
              C_s' & D_s' \\
            \end{array}
          \right)
\left(
  \begin{array}{c}
    \phi'_{1,n} \\
    \phi'_{2,n} \\
  \end{array}
\right),~~s=1,2,\cdots,
\end{equation}
where $A_s', B'_s, C'_s$ and $D_s'$ are polynomials of $\eta$.
On the other side,
if expanding $A_s', B'_s, C'_s$ and $D_s'$  into polynomials of $1/\eta$,
one can derive the so-called  R-Toda$(-)$ hierarchy
\begin{equation}\label{rtoda-hieneg}
    \left(
      \begin{array}{c}
        \ln z'_{n} \\
        \ln r'_{n} \\
      \end{array}
    \right)_{t_{-(s+1)}} = L_{R^-} \left(
      \begin{array}{c}
        \ln z'_{n} \\
        \ln r'_{n} \\
      \end{array}
    \right)_{t_{-s}},~~s=1,2,\cdots,
  \end{equation}
where the recursion operator $L_{R^-}$ is
  \begin{equation}\label{rtoda-hieneg-ope}
L_{R^-}= L_{R^+}^{-1} = \left(
               \begin{array}{cc}
                 \frac{1}{z'_n}+\frac{1}{z'_n} (r'_nE-E^{-1}r'_n )\frac{1}{z'_n}(E+1)\Delta^{-1}
                  & -\frac{1}{z'_n}(r'_nE-E^{-1}r'_n )\frac{1}{z'_n}\Delta^{-1} \\
                 -\Delta \frac{1}{z'_n}(E+1)\Delta^{-1} & \Delta \frac{1}{z'_n}\Delta^{-1} \\
               \end{array}
             \right).
  \end{equation}
Here for the purpose of identification, we have used $(z'_n, r'_n)$ in stead of $(v_n, r_n)$ in the R-Toda$(-)$ hierarchy.
The  first equation in the R-Toda(--) hierarchy is, denoted by the R-Toda(--1),
\begin{equation}\label{rtoda-hieneg1}
\left(
  \begin{array}{c}
    z'_{n,t_{-1}} \\
    r'_{n,t_{-1}} \\
  \end{array}
\right) = \left(
            \begin{array}{c}
              r'_n/z'_{n+1}-r'_{n-1}/z'_{n-1} \\
              r'_{n}/z'_n-r'_n/z'_{n+1} \\
            \end{array}
          \right).
\end{equation}

Ref.\cite{Kharchev-1997-jmpa} used to introduce a transformation
\begin{equation}\label{rtoda-trans+-1}
  z_n\rightarrow \frac{1}{z'_n},~~r_n\rightarrow \frac{r'_n}{z'_n z'_{n+1}},~~ t_1\rightarrow -t_{-1},
\end{equation}
by which the R-Toda(+1) equation \eqref{rtodap-eq} and the R-Toda(--1) equation \eqref{rtoda-hieneg1}
can be transformed to each other.
Next, we will show that the same transformation can be extended to the whole hierarchy of the R-Toda lattice.

\begin{theorem}\label{rtoda-theo-rtoda+-}
  The R-Toda(+) hierarchy \eqref{rtodap-coro-rtoda-hie} and the R-Toda(--) hierarchy \eqref{rtoda-hieneg}
  are equivalent to each other up to the transformation
  \begin{equation}\label{rtoda-theo-rtoda+-trans}
    z_n= \frac{1}{z'_n},~~r_n = \frac{r'_n}{z'_n z'_{n+1}},~~t_s = -t_{-s},~~s=1,2,\cdots.
  \end{equation}
\end{theorem}

\begin{proof}
 The relation \eqref{rtoda-theo-rtoda+-trans} indicates
    \begin{equation}\label{rtoda-theo-rtoda+-trans-1}
      \left(
        \begin{array}{c}
          \ln z_n \\
          \ln r_n \\
        \end{array}
      \right) = T  \left(
        \begin{array}{c}
          \ln z'_n \\
          \ln r'_n \\
        \end{array}
      \right),
      ~~ T=\left(
                  \begin{array}{cc}
                    -1 & 0 \\
                    -E-1 & 1 \\
                  \end{array}
                \right).
    \end{equation}
    Noting that $T^{-1}=T$, and under \eqref{rtoda-theo-rtoda+-trans} there is
    \[L_{R^-}=  T L_{R^+} T,\]
    we can unify the R-Toda$(\pm)$ hierarchies and then complete the proof.
\end{proof}

\subsubsection{The differential-difference Burgers hierarchy}\label{sec-4-3-2}

In the following we will see that one-field reduction of the R-Toda$(+)$ hierarchy can give rise to
the differential-difference Burgers hierarchy.

Imposing reduction
\begin{equation}\label{red-Burgers-1}
r_{n}=-z_{n}
\end{equation}
on the  R-Toda$(+1)$ equation \eqref{rtodap-eq}, we have
\begin{equation}\label{d-Bur-1}
  (\ln z_{n})_{t_1} =W_1= -\Delta z_{n-1} ,
\end{equation}
and on the R-Toda$(+)$ hierarchy \eqref{rtodap-coro-rtoda-hie} we get
\begin{subequations}\label{Bur-hie-1}
\begin{equation}\label{Bur-hie-1-a}
(\ln z_n)_{t_{s+1}}=T_1 (\ln z_n)_{t_{s}}=-T^s_1 \Delta z_{n-1}=W_{s+1},
\end{equation}
where the recursion operator is
\begin{equation}\label{T1}
T_1=-\Delta E^{-1} z_n \Delta^{-1}.
\end{equation}
\end{subequations}
This hierarchy can be explicitly written as
\begin{equation}\label{Bur-hie-1-b}
(\ln z_n)_{t_{s}}=W_s=(-1)^s \Delta \prod^{s}_{j=1} z_{n-j},~~ s=1,2,\cdots.
\end{equation}
Note that by a discrete Cole-Hopf transformation
\begin{equation}\label{Cole-H-1}
z_n=\frac{\alpha_{n}}{\alpha_{n+1}},
\end{equation}
the hierarchy \eqref{Bur-hie-1-b} can be linearized as
\begin{equation}
\alpha_{n,t_s}=(-1)^{s+1}\alpha_{n-s}+c_s(t)\alpha_n,
\end{equation}
where $c_s(t)$ is an arbitrary function of $t$ but independent of $n$.

Equation \eqref{d-Bur-1} is known as the differential-difference Burgers equation.
In fact, taking
\begin{equation}\label{cont-1}
z_n=1+\varepsilon \gamma,~~\partial_{t'_{1}}=\frac{2}{\varepsilon^2}( \partial_{t_1}+ \varepsilon \partial_x)
\end{equation}
and letting $n\to \infty$, $\varepsilon\to 0$ while $n\varepsilon =x$,
\eqref{d-Bur-1} gives rises to
\begin{equation}\label{Bur-1}
\gamma_{t'_1}+2\gamma\gamma_x-\gamma_{xx}=0
\end{equation}
in its leading term, which is the Burgers equation.

To consider the continuum limit of the whole hierarchy \eqref{Bur-hie-1-a}, we introduce
\begin{equation}\label{gamma}
z_n=e^{\varepsilon \gamma}, ~~ x=\varepsilon n.
\end{equation}
Noticing that formally
\begin{subequations}\label{continuum}
  \begin{align}
     & z_n=e^{\varepsilon \gamma}=1+\varepsilon \gamma +O(\varepsilon^2),\\
     & E=e^{\varepsilon \partial_x}=1+\varepsilon \partial_x + O(\varepsilon^2),\\
     & E^{-1}=e^{-\varepsilon \partial_x}=1-\varepsilon \partial_x + O(\varepsilon^2),\\
     & \Delta= (E-1)= \varepsilon \Bigl(\partial_x + \frac{\varepsilon}{2}\partial_x^2 + O(\varepsilon^2)\Bigr),\\
     & \Delta^{-1}= \frac{1}{\varepsilon}\Bigl(\partial^{-1}_x -\frac{\varepsilon}{2} + O(\varepsilon^2)\Bigr),
  \end{align}
\end{subequations}
from \eqref{T1} we have
\begin{equation}
T_1= -1 + \varepsilon \partial_x(\partial_x-\gamma)\partial^{-1}_x +  O(\varepsilon^2).
\end{equation}
This means, after a combination of the flows $W_j$,
\begin{equation}\label{d-Bur-1-a}
(\ln z_{n})_{t_s}=W'_s=(T_1+1)^{s-1}W_1,~~ s=1,2,\cdots,
\end{equation}
can be considered as the differential-difference Burgers hierarchy, as with the continuum limit scheme \eqref{continuum}
it gives rise to the Burgers hierarchy (cf.\cite{Zhang-PS-2011})
\begin{equation}
\gamma_{t'_{s}}=\partial_x(\partial_x-\gamma)^{s-1}\gamma,
\end{equation}
where $\partial_{t_s}$ has been replaced with $-\varepsilon^s\partial_{t'_s}$.

There is another reduction
\begin{equation}\label{red-Bur-2}
r_n=-z_{n+1},
\end{equation}
which leads the R-Toda$(+)$ hierarchy \eqref{rtodap-coro-rtoda-hie} to
\begin{subequations}\label{Bur-hie-2}
\begin{equation}\label{Bur-hie-2-a}
(\ln z_n)_{t_{s+1}}=T_2 (\ln z_n)_{t_{s}}=-T^s_2 \Delta z_{n}=Y_{s+1},
\end{equation}
where the recursion operator is
\begin{equation}\label{T2}
T_2=-\Delta z_n E \Delta^{-1}.
\end{equation}
\end{subequations}
Note that the two hierarchies \eqref{Bur-hie-1} and \eqref{Bur-hie-2}
are simply related by $z_{n+j}\to z_{n-j}$ and $t_s\to -t_s$.
The first equation in the hierarchy in \eqref{Bur-hie-2} reads
\begin{equation}\label{d-Bur-2}
  (\ln z_{n})_{t_1} = Y_1= -\Delta z_{n},
\end{equation}
which, under \eqref{cont-1}, gives the Burgers equation
\begin{equation}\label{Bur-2}
\gamma_{t'_1}+2\gamma\gamma_x+\gamma_{xx}=0.
\end{equation}
The whole hierarchy can also be explicitly written as
\begin{equation}\label{Bur-hie-2-b}
(\ln z_n)_{t_{s}}=Y_s=(-1)^s \Delta \prod^{s-1}_{j=0} z_{n+j},~~ s=1,2,\cdots,
\end{equation}
which, by a discrete Cole-Hopf transformation,
\begin{equation}\label{Cole-H-2}
z_n=\frac{\beta_{n+1}}{\beta_{n}},
\end{equation}
are linearized as
\begin{equation}
\beta_{n,t_s}=(-1)^{s}\beta_{n+s}+c_s(t)\beta_n,
\end{equation}
where $c_s(t)$ is an arbitrary function of $t$ but independent of $n$.
The combined  hierarchy
\begin{equation}\label{d-Bur-2-a}
(\ln z_n)_{t_s}=Y'_s=(T_2+1)^{s-1}Y_1,~~ s=1,2,\cdots
\end{equation}
can be considered as a second differential-difference Burgers hierarchy with continuum limit (cf.\cite{ChenDY-JMP-2002,Olver-1977-jmp,Zhang-PS-2011})
\begin{equation}\label{BH}
\gamma_{t'_{s}}=\partial_x(\partial_x+\gamma)^{s-1}\gamma.
\end{equation}
Note that in the continuum limit scheme \eqref{continuum}
the equation
\begin{equation}\label{Bur-d-2}
(\ln z_n)_{t_2}=\Delta z_nz_{n+1}
\end{equation}
goes to the Burgers equation \eqref{Bur-2} when $n\to +\infty$,
and \eqref{Bur-d-2} was derived in \cite{Levi-Bur-1983} as a discrete Burgers equation.

The derivation of the Burgers hierarchy in \cite{Levi-Bur-1983} indicates that the spectral problem for the continuous Burgers hierarchy
is (also see Eq.(5.2) in \cite{Konopelchenko-rmp-1990})
\begin{equation}\label{sp-Bur}
\psi_x-\gamma \psi =\lambda \psi.
\end{equation}
For the differential-difference Burgers hierarchy \eqref{Bur-hie-1}, following from \eqref{rtoda-spuniversity} and \eqref{red-Burgers-1},
its spectral problem reads
\begin{equation}
  \phi'_{1,n+1} + z_n\phi'_{1,n}=\eta^2(\phi'_{1,n} +z_{n-1} \phi'_{1,n-1}),
\end{equation}
which is gauge equivalent to
\begin{equation}\label{sp-Bur-d1}
\varphi_{n+1}=   \zeta z_n\varphi_{n},
\end{equation}
where
\begin{equation}\label{trans}
\varphi_n=(-\zeta)^{-n}\phi'_n,~~ \zeta=\eta^{-2}.
\end{equation}
Now we take
\begin{equation}
\varphi_n=e^{\varepsilon \psi_n},~~ \zeta=e^{\varepsilon \lambda},
\label{continuum-2}
\end{equation}
by which, together with \eqref{continuum}, the discrete spectral problem \eqref{sp-Bur-d1} gives rise to the continuous spectral problem
\eqref{sp-Bur}.

The reduction \eqref{red-Bur-2} leads to a second discrete spectral problem from \eqref{red-Bur-2}, which is
\begin{equation}\label{sp-Bur-d2}
  \phi'_{1,n+1} + z_n(\phi'_{1,n}-\eta^2 \phi'_{1,n-1})=\eta^2 \phi'_{1,n}.
\end{equation}
Again, employing the transformation \eqref{trans}, together with defining $\varphi'_n=\Delta \varphi_{n-1}$,
we arrive at
$\varphi'_{n+1}=   \zeta z_n\varphi'_{n}$, i.e. \eqref{sp-Bur-d1},
which leads to the spectral problem \eqref{sp-Bur} in continuum limit.

We will have a closer look at the differential-difference Burgers hierarchies in Appendix \ref{app-1}.


\section{The D$\Delta$mKP(E) and constraint}\label{sec-5}

\subsection{The D$\Delta$mKP(E) hierarchy}\label{sec-5-1}

Note that the negative powers of $\Delta$ can be expressed in terms of backward shifts like
\begin{equation}\label{mkp2-hieidEg}
  \Delta^{-s}= \sum_{i=s}^{\infty}\mathrm{C}^{s-1}_{i-1} E^{-i},~~s=1,2,\cdots,
\end{equation}
by which we can rewrite the pseudo-difference operator \eqref{mkp-hie-operator}  as the following
\begin{equation}\label{mkp2hie-operator}
  \bar L= wE +w_0+w_1E^{-1}+\cdots,
\end{equation}
where the new valuables $\{w,w_s\}$ are related to $\{v,v_s\}$ through
\begin{equation}\label{mkp2hie-vw}
  w=v,~~w_0=v_0-v,~~w_s=\sum_{j=1}^{s}\mathrm{ C}_{s-1}^{j-1}\, v_j,~~s=1,2,\cdots.
\end{equation}
Note also that asymptotically
\[w \to 1,~~ w_0\to -1,~~ w_s\to 0,~ (s=1,2,\cdots)
\]
as $|n|\to +\infty$.

The hierarchy resulted from \eqref{mkp2hie-operator} is named as the D$\Delta$mKP(E) hierarchy,
which is generated from the Lax triplet
\begin{subequations}\label{mkp2hie-sp}
  \begin{align}
     & \bar{L}\Theta = \lambda \Theta, \label{mkp2hie-sp-1}\\
     & \Theta_{x} = \bar A_1 \Theta, \label{mkp2hie-sp-x}\\
     & \Theta_{t_s} = \bar A_s \Theta,~~s=1,2,\cdots, \label{mkp2hie-sp-2}
  \end{align}
\end{subequations}
where $\bar A_s=(\bar L^s)_{\geq 1}$, of which the first two of $\bar A_s$ are
\begin{subequations}\label{mkp2hie-C12}
  \begin{align}
     &  \bar A_1=   w E,\\
     &  \bar A_2=   w(Ew)E^2+ ( w(Ew_0)+ww_0)E.
  \end{align}
\end{subequations}
The compatibility of \eqref{mkp2hie-sp} is
\begin{subequations}\label{mkp2-lax}
  \begin{align}
     & \bar{L}_{t_s} = [\bar A_s, \bar L], \label{mkp2-lax-1}\\
     & \bar{L}_{x} = [\bar A_1, \bar L], \label{mkp2-lax-2}\\
     & \bar A_{1,t_s}-\bar A_{s,x}+  [\bar A_1, \bar A_s]=0,~~ s=1,2,\cdots. \label{mkp2-lax-3}
  \end{align}
\end{subequations}
From \eqref{mkp2-lax-2} one can express $w_s$ in terms of $w$ as the following,
\begin{equation}\label{mkp2hie-laxs1v0}
  w_0 = \Delta^{-1} (\ln w)_{x}-1,~~ w_1= \frac{\Delta^{-2}( \ln w)_{xx} }{E^{-1}w},~~
  w_{s+1} = \pi_{s+1}^{-1} \Delta^{-1}\pi_{s} \, w_{s,x},~~s=1,2,\cdots,
\end{equation}
where $\{\pi_s\}$ is defined by
$\pi_s= \prod_{i=1}^{s} (E^{-i}  w)$, for $s=1,2,\cdots$;
and from \eqref{mkp2-lax-3} we have
\begin{equation}\label{mkp2-hie}
 w_{t_s}= \bar K_s=(\bar A_{s,x}-   [\bar A_1, \bar A_s])E^{-1},~~ s=1,2,\cdots,
\end{equation}
which provides a zero curvature expression of the scalar D$\Delta$mKP(E) hierarchy.
An alternative expression of  \eqref{mkp2-hie} is
\begin{equation}\label{mkp2hie-hie}
 w_{t_s} = \bar K_s = w \Delta\underset{E}{\mathrm{ Res}}(\bar{L}^sE^{-1}),~~s=1,2,\cdots.
\end{equation}
The first two equations are
\begin{equation}\label{mkp2hie-hies1}
  w_{t_1} =\bar K_1= w_x
\end{equation}
and
\begin{equation}\label{mkp2hie-hies2}
  (\ln w)_{t_2} = (1+2\Delta^{-1})(\ln w)_{xx} +(\ln w)_{x} (1+2\Delta^{-1})(\ln w)_{x}-2(\ln w)_{x},
\end{equation}
or in the form
\begin{equation}\label{mkp2hie-hies2potential}
  \hat{w}_{t_2} = (1+2\Delta^{-1}) \hat{w}_{xx}  +\hat{w}_{x}(1+2\Delta^{-1}) \hat{w}_{x}-2  \hat{w}_{x}
\end{equation}
with $w=e^{\hat{w}}$.

In addition, similar to \cite{fu-2013-nonlinearity}, we can prove that
\begin{theorem}\label{mkp2-theo-symmetry}
  The $\{S_s\}$ defined by \eqref{mkp2hie-hie} are the infinitely many symmetries of the scalar D$\Delta$mKP(E) hierarchy \eqref{mkp2hie-hie},
  i.e. $[\![S_i,S_j]\!]=0$.
\end{theorem}

\subsection{Squared eigenfunction symmetry}\label{sec-5-2}

Consider the Lax triplet \eqref{mkp2hie-sp} and its adjoint form
\begin{subequations}\label{mkp2hie-sp*}
  \begin{align}
     & \bar{L}^*\Theta^* = \lambda \Theta^*, \label{mkp2hie-sp*-1}\\
     & \Theta^*_{x} = -\bar A^*_1 \Theta^*, \label{mkp2hie-sp*-x}\\
     & \Theta^*_{t_s} = -\bar A^*_s \Theta^*,~~s=1,2,\cdots. \label{mkp2hie-sp*-2}
  \end{align}
\end{subequations}
One can verify that if
 \begin{equation}\label{mkp2sc-spec-laxasso}
   \bar{L}_z =[ -\Theta E\Delta^{-1}\Theta^*,\bar{L}],
 \end{equation}
then $[\partial_z,\partial_{t_s}]\bar L =0$.
\eqref{mkp2sc-spec-laxasso} indicates that
\begin{equation}\label{mkp2sc-spec-asso-wTheta}
  w_z= w(\Delta \Theta \Theta^*),
\end{equation}
which means $\sigma = w(\Delta \Theta \Theta^*)$ is a (squared eigenfunction) symmetry of the whole D$\Delta$mKP(E) hierarchy \eqref{mkp2-hie}.

\subsection{Symmetry constraint}\label{sec-5-3}

\subsubsection{Spectral problem}\label{sec-5-3-1}

Now let us consider a constraint
\begin{equation}\label{mkp2sc-spec-sc}
  w_x=w(\Delta \Theta \Theta^*).
\end{equation}
For convenience we take
$\Theta=c_n,~ \Theta^*=d_n$.
Then, first, compared with $w_0$ in \eqref{mkp2hie-laxs1v0} we find
\begin{equation}\label{mkp2sc-spec-sc-w0}
  w_0= c_n d_n.
\end{equation}
Next, noting that $\Theta$ and $\Theta^*$ satisfy \eqref{mkp2hie-sp-x} and \eqref{mkp2hie-sp*-x},
by calculation from \eqref{mkp2sc-spec-sc} we can find
\begin{equation}\label{mkp2sc-spec-sc-w}
  w= c_n d_{n+1},
\end{equation}
which gives an explicit form of the symmetry constraint  \eqref{mkp2sc-spec-sc}.
Further,  those $w_s$ defined in \eqref{mkp2hie-laxs1v0} can be explicitly expressed as
\begin{equation}\label{mkp2sc-spec-sc-ws}
  w_s=c_n d_{n-s},~~s=1,2,\cdots.
\end{equation}
Then, making use of the formula of $\Delta^{-1}$ in \eqref{mkp2-hieidEg},
we can rewrite $\bar L$ in terms of $(c_n, d_n)$ as the compact form
\begin{equation}\label{mkp2sc-spec-sc-L}
  \bar{L} = c_n E^{2}\Delta^{-1} d_n.
\end{equation}

Thus, the spectral problem \eqref{mkp2hie-sp-1} reads
\begin{equation}\label{mkp2sc-spec-sc-sp}
  c_n E^{2}\Delta^{-1} d_n \Theta= \lambda \Theta,
\end{equation}
which is gauge equivalent to
\begin{equation}\label{mkp2sc-spec-sc-spVol}
  \Theta'_{n+1} - e^{q_n} \Theta'_{n-1}=\xi \Theta'_{n}
\end{equation}
by the transformation
\begin{equation}\label{mkp2sc-spec-sc-gt}
  \Theta= c_n \lambda^{n/2} \left(\prod^{n}_{k=-\infty}\frac{-1}{c_k d_k}\right)\Theta'_{n} ,~~\lambda = \xi^2,
\end{equation}
and
\begin{equation}\label{qn-cndn}
e^{q_n}=-c_nd_n.
\end{equation}
\eqref{mkp2sc-spec-sc-spVol} is nothing but the spectral problem of the Volterra lattice hierarchy.


\subsubsection{Decomposition of the Volterra hierarchy}\label{sec-5-3-2}

With new valuables $c_n$ and $d_n$, \eqref{mkp2hie-sp-2} and \eqref{mkp2hie-sp*-2} read
\begin{equation}\label{mkp2schie-abts}
  c_{n,t_s} = \bar A_s c_n,~~d_{n,t_s}= -\bar A^*_s d_n,~~s=1,2,\cdots,
\end{equation}
where $\bar A_s= ((c_n E^{2}\Delta^{-1} d_n)^s)_{\geq 1}$.
Let us see what the known integrable hierarchy is related to the above system.

First, \eqref{mkp2hie-sp-x} and \eqref{mkp2hie-sp*-x} are
\begin{equation}\label{mkp2schie-abt1}
  c_{n,x} = c_nc_{n+1}d_{n+1},~~d_{n,x}= -d_nd_{n-1}c_{n-1},
\end{equation}
which gives rise to the Volterra equation
\begin{equation}\label{mkp2schie-Voleq}
  q_{n,x}=- e^{q_{n+1}} + e^{q_{n-1}}
\end{equation}
provided $q_n$ is defined through \eqref{qn-cndn}.
In order to understand the recursive structure behind \eqref{mkp2schie-abts}, we need
the following relations
  \begin{subequations}
    \begin{align}
       & \bar A_s c_n E\Delta^{-1} d_n = (\bar A_s c_n E\Delta^{-1} d_n )_{\geq 1} + (\bar A_s c_n) E\Delta^{-1} d_n , \\
       & c_n E\Delta^{-1} d_n \bar A_s =  (c_n E\Delta^{-1} d_n \bar A_s )_{\geq 1} + c_n E\Delta^{-1} (\bar A^*_s d_n),
    \end{align}
  \end{subequations}
which holds by considering the definition of $\bar A_s$ and the formula of $ \Delta^{-1}$  in \eqref{mkp2-hieidEg}.
Then we have  the recursive structure for $\bar A_s$ and $\bar A^*_s$:
\begin{subequations}\label{mkp2schie-Csrecursion}
  \begin{align}
     & \bar A_{s+1}    = c_n E d_n ( \bar A_s+ ( \bar L^s)_0 ) + c_nE\Delta^{-1} d_n\bar A_s- c_n E\Delta^{-1} (\bar A^*_s d_n ),\\
     & \bar A^*_{s+1} = d_n E^{-1}c_n ( \bar A^*_s + (\bar{L}^s)_0) - d_n \Delta^{-1}c_n \bar A^*_s + d_n \Delta^{-1} (\bar A_s c_n),
  \end{align}
\end{subequations}
in which the term $(\bar{L}^s)_0$ is written as
\begin{equation}\label{mkp2schie-wabts}
  (\bar{L}^s)_0 = \Delta^{-1} (\ln c_n)_{t_s} + \Delta^{-1}E (\ln d_n)_{t_s},~~s=1,2,\cdots
\end{equation}
by noting that \eqref{mkp2sc-spec-sc-w} and \eqref{mkp2hie-hie}.
Then, with the help of the above relations,  it is not difficult to prove that
\begin{theorem}\label{mkp2schie-theo1}
  \eqref{mkp2schie-abts} can be written as the recursive form
  \begin{equation}\label{mkp2schie-theo1-recursion}
    \left(
      \begin{array}{c}
        \ln c_n \\
        \ln d_n \\
      \end{array}
    \right)_{t_s} = T^{s-1} \left(
                                       \begin{array}{c}
                                         c_{n+1}d_{n+1} \\
                                         -d_{n-1}c_{n-1} \\
                                       \end{array}
                                     \right),~~s=1,2,\cdots,
  \end{equation}
  where the recursion operator $T$ is defined as
  \begin{equation}\label{mkp2schie-theo1-recursion-oper}
    T= \left(
                      \begin{array}{cc}
                        E & 0 \\
                        0 & -E^{-1} \\
                      \end{array}
                    \right) \Delta^{-1} (Ec_nd_nE- c_nd_n )\Delta^{-1} \left(
                                                         \begin{array}{cc}
                                                           1 & 1 \\
                                                           1 & 1 \\
                                                         \end{array}
                                                       \right).
  \end{equation}
\end{theorem}

Further we have
\begin{theorem}\label{mkp2schie-coro1}
  By defining $q_n$ as in \eqref{qn-cndn}, the hierarchy \eqref{mkp2schie-theo1-recursion} gives rise to the Volterra hierarchy
  \begin{equation}\label{mkp2schie-coro1-Volhie}
    q_{n,t_s} = (-1)^s L_{V}^{s-1} (e^{q_{n+1}} -e^{q_{n-1}}  ),~~s=1,2,\cdots,
  \end{equation}
  where the recursion operator $L_{V}$ is given in \eqref{Volterra-hie}.
  In this context, we say that \eqref{mkp2schie-theo1-recursion} provides a decomposition of the Volterra hierarchy \eqref{mkp2schie-coro1-Volhie}.
\end{theorem}

Note that if we look for $c_nd_{n+j}=-1$ type of reduction, we find the only available case is $j=-2$, i.e.
\begin{equation}\label{mkp2schie-new-redu}
 d_{n}=\frac{1}{d_{n-2}}.
\end{equation}
This leads  \eqref{mkp2schie-abt1} to a single equation
\begin{equation}\label{mkp2schie-new-equ}
 (\ln c_{n})_{x}= -\frac{c_n}{c_{n-2}},
\end{equation}
which is the the  Volterra  equation \eqref{mkp2schie-Voleq} if we take
\begin{equation}
q_n=\ln \frac{c_n}{c_{n-2}},~~(\mathrm{or}~ e^{q_n}=-c_nd_n,~ d_n={-1}/{c_{n-2}}).
\end{equation}
It is interesting that the reduction \eqref{mkp2schie-new-redu} is  valid as well for the whole hierarchy \eqref{mkp2schie-theo1-recursion}.
As a result, we obtain a scalar hierarchy
\[(\ln c_n)_{t_{s+1}}=-E\Delta^{-1}\Bigl(E \frac{c_n}{c_{n-2}} E - \frac{c_n}{c_{n-2}}\Bigr)(E^{-1}+E^{-2})(\ln c_n)_{t_{s}},~~s=1,2,\cdots,\]
which is, by acting $(1-E^{-2})$ on both sides,
\begin{equation}
    q_{n,t_{s+1}} = - L_{V} q_{n,t_{s}},~~s=1,2,\cdots,
\end{equation}
i.e., the Volterra hierarchy where $L_V$ is given in \eqref{Volterra-hie}.

\subsubsection{Correlation of the D$\Delta$KP and D$\Delta$mKP(E)}\label{sec-5-3-3}

Note that both the D$\Delta$KP and D$\Delta$mKP(E) give rise to the Volterra hierarchy via there
squared eigenfunction symmetry constraints,
there should have some correlations behind the fact.

Let us consider  the gauge transformation
   \begin{equation}\label{mkpequi-theo-equi-gauge}
     h_n \bar{L} = \bar{M} h_n
   \end{equation}
   with an unfixed scalar function $h_n$, which gives rise to
   \begin{equation}\label{mkpequi-theo-equi-relation}
     h_n w=h_{n+1},~~w_0= \bar{u},~~h_n w_s=\bar{u}_s h_{n-s},~~s=1,2,\cdots.
   \end{equation}
From  \eqref{mkp2hie-hie} we know that
  \begin{equation}\label{mkpequi-theo-equi-htime}
    (\ln h_n)_{t_s} =(\bar{L}^s)_0.
  \end{equation}
Then, under \eqref{mkpequi-theo-equi-gauge}, \eqref{mkpequi-theo-equi-htime} and the Lax equation \eqref{mkp2-lax-1}
we find
     \begin{align*}
       \bar{M}_{t_s} &= (h_n \bar{L} h_n^{-1} )_{t_s} = h_{n,t_s}\bar{L} h_n^{-1}+ h_n \bar{L}_{t_s} h_n^{-1} - h_n \bar{L} h_n^{-1} h_{n,t_s} h_n^{-1} \\
        & = h_{n,t_s}h_n^{-1}\bar{M}+ h_n [ \bar A_s,\bar{L} ] h_n^{-1} - \bar{M} h_{n,t_s} h_n^{-1} \\
        & = [h_{n,t_s}h_n^{-1}+ h_n\bar A_sh_n^{-1}, \bar{M}] \\
        & = [h_n (\bar{L}^s)_0h_n^{-1}+ h_n(h_n^{-1}\bar{M}^s h_n)_{\geq 1}h_n^{-1}, \bar{M}]\\
        & = [\bar{B}_s,\bar{M}].
     \end{align*}
This indicates that the Lax equation $\bar{M}_{t_s}= [\bar{B}_s,\bar{M}]$
is a consequence of  \eqref{mkp2-lax-1} if $h_n$ provides the gauge transformation \eqref{mkpequi-theo-equi-gauge}.
Note that $\bar{M}_{t_s}= [\bar{B}_s,\bar{M}]$ and \eqref{gt-kp-lax} generate the same D$\Delta$KP hierarchy,
and with \eqref{mkpequi-theo-equi-gauge} we have
\[\Psi=h_n\Theta_n,~~ \Psi^*=\Theta^*/h_n.\]
By the replacement we employed previously, i.e.
$\Psi=Q_n, ~\Psi^*=R_n,~\Theta=c_n,~ \Theta^*=d_n$,
we formally have
\[Q_n=h_n c_n,~ R_n=d_n/h_n,\]
which gives rise to $Q_nR_n=c_nd_n$. Thus \eqref{vrt-conn} and  \eqref{qn-cndn} coincide,
and it is then not surprised that both the D$\Delta$KP and D$\Delta$mKP(E)
are related to the Volterra hierarchy by squared eigenfunction symmetry constraints.

Note that some connections between the D$\Delta$KP  and D$\Delta$mKP($\Delta$) hierarchy and their
squared eigenfunction symmetry constrained systems have been investigated, e.g. \cite{cheng-2018,Tamizhmani-2000},
but similar connections between D$\Delta$KP  and D$\Delta$mKP(E) have not been reported.

\section{Concluding remarks}\label{sec-6}

We have mainly considered squared eigenfunction symmetry constraint of the D$\Delta$mKP system together with the D$\Delta$KP systems.
These two systems are entangled each other, as we can see in the paper.
As the new results, we have achieved the following.
The D$\Delta$KP gives rise to the Ragnisco-Tu hierarchy which is a discretization of the AKNS system
and is reduced to the Volterra hierarchy.
Note that it used to be thought that there does not exist one-field reduction for the Ragnisco-Tu hierarchy.
The D$\Delta$mKP gives rise to the R-Toda hierarchy which is reduced to the differential-difference Burgers hierarchy.
The D$\Delta$mKP(E) leads to the Volterra hierarchy  as well, by either decomposition or reduction.
In some reductions, we have taken into account of nonzero asymptotic conditions of the wave functions
so that reductions can be implemented reasonably.

The D$\Delta$KP and D$\Delta$mKP  systems can be formally considered as subsystems of the 2DTL system,
since the 2DTL system employs both pseudo-difference operators $\bar M$ and $\bar L^*$ \cite{UT-2DTL-1984}.
There exists a constraint imposed on $\bar M$ and $\bar L^*$ such that
the time part in the Lax triplet of the 2DTL together with its adjoint form
gives rise to the Ablowitz-Ladik hierarchy (cf. \cite{Taka-JPA-2018});
but that is different from the constraint of squared eigenfunction symmetry.
Besides, there exists a local transformation to bring the Ablowitz-Ladik hierarchy to the R-Toda hierarchy \eqref{rtodap-theo-hie-s},
which is given by \cite{Kharchev-1997-jmpa}, (also see \S 18.11 of \cite{suris-2003}),
\begin{equation}
a_n=-\frac{Q_n}{R_{n-1}},~~b_n=-\frac{Q_n}{R_{n-1}}+Q_nR_{n-1}.
\end{equation}
In addition, $(a_n, b_n)$ can also be locally expressed by $(u_n, v_n)$ of the Ragnisco-Tu hierarchy via \cite{Kharchev-1997-jmpa}
\begin{equation}
a_n=u_nv_n+\frac{u_n}{v_{n-1}},~~b_n=-\frac{u_n}{v_{n-1}}.
\end{equation}
However, the connection between the  Ragnisco-Tu and Ablowitz-Ladik potentials,
\begin{equation}
u_n=Q_n h_n, ~~ v_n=\frac{R_{n-1}}{h_n},~~h_n=\prod^{n-1}_{j=-\infty}(1-Q_jR_j),
\end{equation}
are not local.
Although the above three hierarchies are connected via transformations, they are not essentially equivalent,
for example, they exhibit quite different one-field reduction features.

Note that in continuous case  the KP system generates the AKNS spectral problem and its hierarchy \cite{cheng-pla-1991,Konopelchenko-pla-1991}
and the mKP system leads to the  Kaup-Newell spectral problem \cite{kaup-1978} and its hierarchy
that gives rise to the Burgers hierarchy after a simple reduction \cite{ChenDY-JMP-2002}.
The result of the present paper implies possible connection
between the squared eigenfunction symmetry constraint of the D$\Delta$mKP
and the Kaup-Newell hierarchy, which is worthy to investigate elsewhere.
In addition, the differential-difference case exhibits richer results and reveals more links than the continuous case,
some of which already emerged in \cite{Kharchev-1997-jmpa,oevel-1996},
but here have been much more elaborated.
Investigation of continuum limits of the D$\Delta$mKP hierarchies, their integrable characteristics, and
constrained systems will provide more insight to understand the D$\Delta$mKP system.
We believe that the results we obtained in the paper will provide a deeper understanding
for the connections of (1+1)-dimensional and (2+1)-dimensional differential-difference integrable systems.

In Appendix we will have a close look at the  differential-difference Burgers hierarchy and a
fully discrete Burgers equation.
It turns out that the former act as auto B\"acklund transformations for the continuous Burgers hierarchy
and the nonlinear superposition formula of the B\"acklund transformation
gives rise to a fully discrete Burgers equation
that is defined on 3 points (not quadrilateral) but still consistent around  cube.
Moreover,  this discrete Burgers equation is linearisable (see \eqref{dB-linear}).

\subsection*{Acknowledgments}
DJZ thanks Prof. Hietarinta for his discussion on the discrete Burgers equation.
This work was supported by  the NSF of China (Nos. 11631007, 11875040,  11601312).

\appendix
\section{A close look at the Burgers hierarchies}\label{app-1}

Since the squared eigenfunction symmetry of the continuous mKP gives rise to the derivative Schr\"odinger hierarchy
hierarchy that is easily reduced to the Burgers hierarchy \cite{ChenDY-JMP-2002},
it is not surprising the D$\Delta$mKP generates the differential-difference Burgers hierarchy.
However, the differential-difference Burgers equation can not only act as a B\"acklund transformation
for the whole continuous Burgers hierarchy,
but also generate a fully discrete Burgers equation from its nonlinear superposition formula.
Thus, we may have a full profile from the continuous to the discrete Burgers equations.

\subsection{Connections with the KP and mKP system}\label{sec-A-1}

Let us briefly recall the known connections between the Burgers hierarchy and the KP and mKP system formulated in
\cite{Harada-JPSJ-1985,Harada-JPSJ-1987,Kiso-PTP-1990}.
First, in Sato's KP theory the pseudo-differential operator
\begin{equation}\label{L-KP}
L^{}_{\mathrm{KP}}=\partial_x+u_2\partial^{-1}_x +u_3\partial^{-2}_x +\cdots
\end{equation}
is introduced by $L^{}_{\mathrm{KP}}=W \partial_x W^{-1}$ where
\begin{equation}\label{W}
W=1+ w_1 \partial^{-1}_x +w_2\partial^{-2}_x +\cdots.
\end{equation}
The Burgers hierarchy \eqref{BH} is equivalent to the Sato equation \cite{Harada-JPSJ-1985}
\begin{equation}\label{Sato-eq-KP}
W_{t_s}=B_s W -W \partial^s_x,~~ s=1,2,\cdots,
\end{equation}
where $W$ is truncated as
\begin{equation}\label{W-Bur}
W=1+ w_1 \partial^{-1}_x,~~ w_1=-\gamma,
\end{equation}
and $B_s=(W \partial_x^s W^{-1})_{\geq 0}$.
In this sense, the Burgers hierarchy is viewed as a sub-hierarchy of the KP equation.
Note that in terms of the $\tau$ function in Sato's KP theory (cf.\cite{OhtSTT-PTPS-1988}) there is
\begin{equation}\label{w1-tau}
\gamma=-w_1=\tau_x/\tau,
\end{equation}
which leads to the linearisation of the Burgers hierarchy.

Second, the pseudo-differential operator of the mKP hierarchy is
\begin{equation}\label{L-mKP}
L^{}_{\mathrm{mKP}}=\partial_x+v+v_2\partial^{-1}_x +v_3\partial^{-2}_x +\cdots.
\end{equation}
The symmetry constraint
\begin{equation}\label{mKP-con}
v=q q^*
\end{equation}
gives rise to the Kaup-Newell spectral problem and its hierarchy \cite{ChenDY-JMP-2002},
where $q$ is the eigenfunction of the $L^{}_{\mathrm{mKP}}$ and $q^*$ is its adjoint counterpart.
This constraint reduces $L^{}_{\mathrm{mKP}}$ to
\begin{equation}
L^{}_{\mathrm{mKP}}=\partial_x+q\partial^{-1}_xq^*\partial_x
\end{equation}
and the Burgers hierarchy is obtained from a single eigenfunction constraint
$v=q$ (i.e. $q^*=1$), which leads to
\begin{equation}
L=\partial_x+v,
\end{equation}
and $Lq=\lambda q$ is nothing but the spectral problem of the
Burgers hierarchy (cf. Sec.IV in \cite{ChenDY-JMP-2002}).
Note that the two operators  $L^{}_{\mathrm{KP}}$ and  $L^{}_{\mathrm{mKP}}$
are connected by gauge relation (cf.\cite{Kiso-PTP-1990})
$L^{}_{\mathrm{KP}}= f^{-1}L^{}_{\mathrm{mKP}} f$
where $v=f_x/f$,  which coincides with \eqref{w1-tau}
and indicates $v$ and $\gamma$ in \eqref{w1-tau} can share same $\tau$ function.

In the differential-difference case, the pseudo-difference operator $M$, \eqref{gt-kp-operator},
is introduced by $M=W\Delta W^{-1}$ where \cite{kanaga-1997}
\begin{equation}\label{W-d}
W=1+ w_1 \Delta^{-1}  +w_2\Delta^{-2} +\cdots.
\end{equation}
The corresponding  Sato equation reads
\begin{equation}\label{Sato-eq-dKP}
W_{t_s}=B_s W -W \Delta^s,~~ s=1,2,\cdots,
\end{equation}
where $B_s=(W \Delta^s W^{-1})_{\geq 0}$.
One can find that when $W$ is truncated as
\begin{equation}\label{W-Bur-d}
W=1+ w_1 \Delta^{-1}_x,~~ w_1=z_n+1,
\end{equation}
the first two Sato equations are
\[(\ln z_n)_{t_1}=Y_1, ~~ (\ln z_n)_{t_2}=-2Y_1+Y_2, \]
where $Y_j$ are the differential-difference Burgers flows given in \eqref{Bur-hie-2-b}.
More precise relation between the differential-difference Burgers hierarchy and the Sato equation \eqref{Sato-eq-dKP}
with \eqref{W-Bur-d} will be investigated elsewhere.

Finally, back to the squared eigenfunction symmetry constraint \eqref{sc-mkp}
of the D$\Delta$mKP hierarchy,
the differential-difference Burgers hierarchy is obtained when $r_n=-z_{n+1}$ (see \eqref{red-Bur-2}),
which holds when taking $b_n=1$ in light of \eqref{rtoda-abgt1abqr}.
This leads to once again a single eigenfunction constraint $v=a_n=-z_n$ and the constrained operator
\eqref{rtoda-abL} reduces to $L=-z_nE$ and the related spectral problem reads
$z_n\psi_{n+1}=-\lambda \psi_n$, which is a spectral problem of the differential-difference Burgers hierarchy
\eqref{Bur-hie-2-b}.
It can be easily transformed to a trivial form $\beta_{n+1}=z_n\beta_n$ by $\psi_n=(-\lambda)^n/\beta_n$,
which coincides with \eqref{Cole-H-2}.
Note also that the gauge relation \eqref{gt-kp-gaugeLM} (also see \cite{Tamizhmani-2000})
gives the relation $v=-z_n=f_{n+1}/f_n$ in \eqref{gt-kp-vu-a},
but in terms of $\tau$ function $w_1$ in \eqref{W-Bur-d} is expressed as
$w_1=(\ln \tau_n)_x+c$ (cf.\cite{kanaga-1997,Kajiwara-jmp-1991}) rather than $w_1=\tau_{n+1}/\tau_n+c$,
where $x$ is an auxiliary independent variable and $c$ is some constant.

\subsection{B\"acklund and Galilean transformations of the Burgers hierarchy}\label{sec-A-2}

Equation \eqref{mkp-hie-spec-x} in the Lax triplet together with its adjoint form \eqref{gt-kp-adwave-mkp-x}
leads to a differential-difference Burgers equation
\begin{equation}\label{be-be}
  \partial_{x}\ln z_{n}=z_{n-1}-z_{n}.
\end{equation}
Let us look at the relation between \eqref{be-be} together with the  differential-difference Burgers hierarchy
\begin{equation}\label{be-hie}
 \partial_{t_{s+1}}\ln z_{n}=-\Delta E^{-1}z_n\Delta^{-1}\partial_{t_{s}}\ln z_{n},~~\partial_{t_1}\ln z_{n}=z_{n-1}-z_{n},~~s=1,2,\cdots,
\end{equation}
and the Burgers hierarchy w.r.t. $z_n$
\begin{equation}\label{be-hie-con}
   \partial_{t_{s+1}} z_n= -\partial_{x}(\partial_{x}+ z_n)\partial_{x}^{-1} \partial_{t_s}z_n,~~ \partial_{t_1}z_n=z_{n,x}.
\end{equation}

\begin{lemma}\label{lem-A1}
Under \eqref{be-be} the  differential-difference Burgers hierarchy \eqref{be-hie} gives rise to the
Burgers hierarchy \eqref{be-hie-con}.
\end{lemma}
\begin{proof}
As a single equation, \eqref{be-hie} is an $(s+1)$-order difference equation.
In principle, one can repeatedly use \eqref{be-be} to express $z_{n-j}$ in terms of $z_n$, e.g.
\begin{equation*}
  z_{n-1}=\frac{z_{n,x}+z_n^2}{z_n},~~z_{n-2}=\frac{z_n^3+3z_nz_{n,x}+z_{n,xx}}{z_n^2+z_{n,x}},~~\cdots,
\end{equation*}
to eliminate all $z_{n-j}$ in \eqref{be-hie}.
In practice, rewriting \eqref{be-hie} as
\begin{equation}\label{be-hie-alt}
     \partial_{t_{s+1}} E \Delta^{-1}\ln z_{n}=-z_n\partial_{t_{s}}\Delta^{-1}\ln z_{n},
\end{equation}
and replacing the $\Delta^{-1}$ term by using \eqref{be-be} which indicates $\Delta^{-1}\ln z_n=-\partial_{x}^{-1}z_{n-1}$,
from \eqref{be-hie-alt} we have
\[\partial_{t_{s+1}} z_n  =-\partial_{x}z_n \partial_{t_s}\partial_{x}^{-1}z_{n-1}.
\]
Then, replacing $z_{n-1}$ by \eqref{be-be} we have
\begin{align*}
      \partial_{t_{s+1}} z_n & = -\partial_{x}z_n \partial_{t_s}\partial_{x}^{-1}( z_n+\partial_{x}\ln z_n) \\
                             & = -\partial_{x}(z_n +\partial_{x})\partial_{x}^{-1} \partial_{t_s} z_n,
\end{align*}
which is the Burgers hierarchy \eqref{be-hie-con}.
\end{proof}

\begin{lemma}\label{lem-A-2}
If $z_n$ satisfies \eqref{be-hie-con},
then $z_{n-1}$ defined by \eqref{be-be} gives rise to the Burgers hierarchy as well.
\end{lemma}
\begin{proof}
First, making use of \eqref{be-be} and \eqref{be-hie-con}, one has
\begin{align*}
      \partial_{t_{s+1}}z_{n-1} &= \partial_{t_{s+1}}(\partial_{x} \ln z_n+z_n ) \\
       & = (\partial_{x}\frac{1}{z_n}+1 )\partial_{t_{s+1}}z_{n}\\
       & =- (\partial_{x}\frac{1}{z_n}+1 )\partial_{x}(z_n +\partial_{x})\partial_{x}^{-1} \partial_{t_s}z_n.
\end{align*}
Meanwhile,
\begin{align*}
\partial_{x}(z_{n-1} +\partial_{x})\partial_{x}^{-1} \partial_{t_s}z_{n-1}
&=  \partial_{x}(\frac{z_{n,{x}}}{z_n}+z_n)\partial_{x}^{-1}(\partial_{x}\frac{1}{z_n}+1 )\partial_{t_{s}}z_{n}\\
& =\partial_{x}( \frac{z_{n,x}}{z_n}+z_n)(\frac{1}{z_n}\partial_{x}+1 )\partial_{x}^{-1} \partial_{t_s}z_n.
\end{align*}
By direct calculation we can find the above two formulas are same and we then arrive at
\begin{equation}\label{be-hie-con-n-1}
      \partial_{t_{s+1}}z_{n-1}=-\partial_{x}(z_{n-1}+\partial_{x})\partial_{x}^{-1}\partial_{t_s}z_{n-1},
\end{equation}
which is the Burgers hierarchy for $z_{n-1}$.
\end{proof}

On the basis of the above two lemmas, we immediately find the following.

\begin{theorem}\label{The-A-1}
\eqref{be-be} and \eqref{be-hie} define an auto B\"acklund transformation
for the Burgers hierarchy from \eqref{be-hie-con-n-1} to \eqref{be-hie-con}.
\end{theorem}

Note that \eqref{be-be} and \eqref{be-hie} are parameter free.
Let us introduce a parameter $\lambda$ by replacing $z_n$ with $z_n+\lambda$.
It follows from \eqref{be-be} and \eqref{be-hie} that
\begin{subequations}\label{be-be-lamb}
\begin{align}
& z_{n,x}=(z_n+\lambda) (z_{n-1}-z_{n}),\label{be-be-lamb-a}\\
& \partial_{t_{s+1}}\ln (z_n+\lambda)=(-1)^{s+1}\Delta [E^{-1}(z_n+\lambda)]^s z_{n-1},
\end{align}
\end{subequations}
which, by repeating the proof of Lemma \ref{lem-A1} and \ref{lem-A-2},
compose a B\"acklund transformation with parameter $\lambda$ for the Burgers hierarchy
\begin{equation}\label{be-hie-con-lamb}
   \partial_{t_{s+1}} z_n= (-1)^s \partial_{x}(\partial_{x}+ z_n+ \lambda)^s  z_{n,x}.
\end{equation}
This is a linear combination of the original Burgers flows in \eqref{be-hie-con},
but also as a whole it is a Galilean transformation of the $(s+1)$-order Burgers equation (see Theorem 1 in \cite{Zhang-PS-2011}).
Thus we can conclude that
\begin{theorem}\label{The-A-2}
\eqref{be-be-lamb} defines an auto B\"acklund transformation
for the Galilean transformed Burgers hierarchy \eqref{be-hie-con-lamb}.
\end{theorem}

\subsection{Nonlinear superposition formula and the discrete Burgers equation}\label{sec-A-3}

There is a Bianchi identity for the Burgers hierarchy.
Rewrite the shifted B\"acklund transformation \eqref{be-be-lamb-a} into
\[\t z_x=(\t z-p)(z-\t z),~~ \h z_x=(\h z-q)(z-\h z),\]
where $\t z$ and $\h z$ denote new solutions by choosing $\lambda=-p$ and $-q$ respectively,
from which Ref.\cite{Levi-Bur-1983} derived a superposition formula
\begin{equation}
\th z=\frac{p \h z-q \t z}{p-q+\h z-\t z}.
\label{NSF}
\end{equation}
Suppose that $z_{n,m}$ is a function defined on $\mathbb Z^2$ lattice as depicted in Fig.\ref{F-1} (a),
introduce notations
\begin{equation*}
z\equiv z_{n,m},~~ \widetilde{z}\equiv z_{n+1,m},~~\widehat{z}\equiv z_{n,m+1},
~~\widehat{\widetilde{z}}\equiv z_{n+1,m+1},
\end{equation*}
and $p,q$ serve as lattice parameters associated with $n$, $m$ directions, respectively.
Then the Bianchi identity \eqref{NSF} can be viewed as a discrete Burgers equation.
In fact, introducing $x=m/q$, in the limit $m,q\to \infty$, \eqref{NSF} gives rise to \eqref{be-be-lamb-a}
which is a differential-difference Burgers equation.
Note that \eqref{NSF} is different from the one derived in \cite{Hernandez-JPA-1999} by discretising Lax pair.
\begin{figure}[h!]
\begin{center}
\begin{tikzpicture}[rotate=0]
\fill \foreach \p in {(0.0,-1.2), (2.0, -1.20), (2.0, 0.8), (0.0, 0.8)} {\p circle (2pt)};
\fill[lightgray] (0,0.8) -- (2,-1.2) -- (2,0.8) -- cycle;
\draw[thick]     (-1.2, -1.2)--(3.2, -1.2) (-0.0, 2.0)--(0.0, -2.4) (2.0, 2.0)--(2.0, -2.4) (-1.2, 0.8)--(3.2, 0.8);
\draw[ thin,color=gray] (7,-1.20)--(10.0,-1.20)  (7,-1.20)--(7.0,1.80) (7,-1.20)--(5.8,-2.4);
\node[above] at (0,2) {$m$};
\node[right] at (3.2, -1.2) {$n$};
\node[left] at (0, -0.3) {$q$};
\node[above] at (1, -1.2) {$p$};
\node[left] at (0, -1.4) {$z$};
\node[left] at (0, 1.1) {$\h z$};
\node[right] at (2.0, -1.45) {$\t z$};
\node[right] at (2.0, 1.1) {$\th z$};
\fill[lightgray] (6.3,-1.9) -- (9,-1.2) -- (8.3,-1.9) -- cycle;
\fill[lightgray] (6.3,-1.9) -- (6.3,0.1) -- (7.0,0.8) -- cycle;
\fill[lightgray] (9,0.8) -- (9,-1.2) -- (7.0,0.8) -- cycle;
\fill[green] (6.3,0.1) -- (8.3,0.1) -- (8.3,-1.9) -- cycle;
\fill[yellow] (9,0.8) -- (8.3,0.1) -- (8.3,-1.9) -- cycle;
\fill[pink] (6.3,0.1) -- (8.3,0.1) -- (9,0.8) -- cycle;
\draw[thick]
(7,0.8) -- (9,0.8) -- (9,-1.2) -- (8.3,-1.9) -- (6.3,-1.9)
(9,0.8) -- (8.3,0.1)-- (8.3,-1.9)
(8.3,0.1)--(6.3,0.1)--(6.3,-1.9)
(6.3,0.1)--(7,0.8);
\node[above] at (7,1.8) {$l$};
\node[right] at (10, -1.2) {$m$};
\node[left] at (5.8,-2.5) {$n$};
\node[left] at (7, -1.2) {$z$};
\node[left] at (6.3, -1.85) {$\t z$};
\node[above] at (9.2, -1.2) {$\h z$};
\node[below] at (8.3, -1.9) {$\th z$};
\node[left] at (7, 0.85) {$\b z$};
\node[left] at (6.3, 0.1) {$\t{\b z}$};
\node[right] at (9, 0.8) {$\h{\b z}$};
\node[above] at (8.3, 0.1) {$\th{\b z}$};
\node[right] at (7, -2.2) {$q$};
\node[above] at (8.5, -1.1) {$r$};
\node[right] at (8.6, -1.7) {$p$};
\node[below] at (1, -2.5) {(a)};
\node[below] at (8, -2.5) {(b)};
\end{tikzpicture}
\caption{\label{F-1} (a). Discrete Burgers equation \eqref{NSF} on $(n,m)$-lattice.
(b). Consistency-around-cube of the discrete Burgers equation.}
\end{center}
\end{figure}
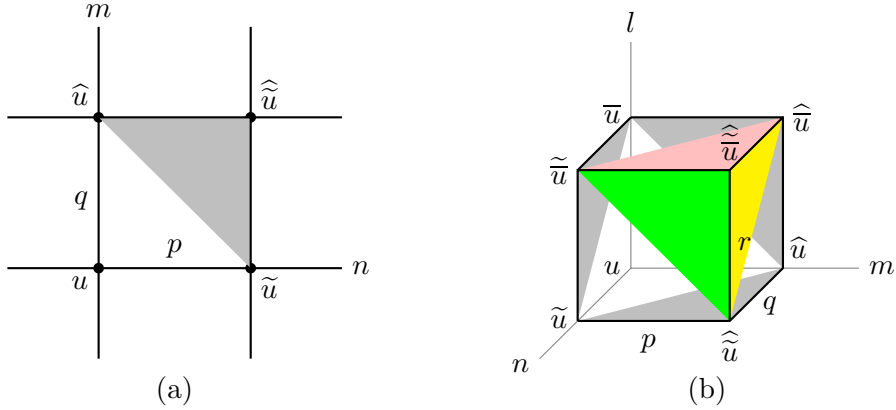

It is remarkable that this 3-point equation \eqref{NSF} can be consistently embedded into a multidimensional lattice
(see Fig.\ref{F-1}(b) where $r$ is the spacing parameter of the bar-direction)
and the value of $\th{\b z}$  is uniquely given as $\th{\b z}=\frac{A}{B}$ where
\begin{align*}
A&= (q-r)(qr\t z+p\h z\b z)+(r-p)(rp\h z+q\b z\t z)+(p-q)(pq\b z+r\t z\h z),\\
B&= p^2(q-r)+q^2(r-p)+r^2(p-q)-(\t z-\h z)(pq+r^2-r\b z)\\
  &~~~~  - (\h z-\b z)(qr+p^2-p \t z)- (\b z-\t z)(pr+q^2-q \h z).
\end{align*}
The consistency-around-cube implies existence of Lax pair (cf.\cite{ABS,Bri-FCM-2013,FNij-2002}).
For that we have
\begin{align}
\t\Phi=U\Phi=\left(\begin{array}{cc}
              p & -r\t z\\
              1 & p-r-\t z
              \end{array}\right)\Phi,~~~
\h\Phi=V\Phi=\left(\begin{array}{cc}
              q & -r\h z\\
              1 & q-r-\h z
              \end{array}\right)\Phi,
\end{align}
and the compatibility $\th\Phi=\t{\h \Phi}$, i.e., $\h U V=\t V U$ gives rise to the discrete Burgers equation \eqref{NSF}.
Note that \eqref{NSF} does not belong to the Adler-Bonenko-Suris classification in \cite{ABS} as it is not quadrilateral.

There is an alternative (trivial) Lax pair (without spectral parameter $r$), which is
\begin{subequations}\label{Lax-dB}
\begin{align}
& \t \varphi =(p-\t z)\varphi,\label{Lax-a}\\
& \h \varphi =(q-\h z)\varphi,\label{Lax-b}
\end{align}
\end{subequations}
from which one can easily check that the compatibility of $\h{\t\varphi}=\t{\h\varphi}$
gives rise to the discrete Burgers equation \eqref{NSF}.
The equation \eqref{Lax-a} indicates $\t z=p-{\psi}/{\t\psi }$, i.e.
\begin{equation}\label{u-psi}
z_{n,m}=p-\frac{\psi_{n-1,m}}{\psi_{n,m} }
\end{equation}
where we have taken
\begin{equation}\label{phi-psi}
\psi=1/\varphi.
\end{equation}
It follows that if $z$ is defined as \eqref{u-psi}, then \eqref{Lax-b} is equivalent to a linear equation
(in terms of $\psi$)
\begin{equation}\label{dB-linear}
\h \psi-\t \psi=(p-q) \th\psi.
\end{equation}
This means, if $\psi$ is a solution to \eqref{dB-linear}, then $z$ defined by \eqref{u-psi}
provides a solution to the discrete Burgers equation \eqref{NSF}.


\end{document}